\documentclass[11pt]{article}
\usepackage[pdftex]{graphicx}
\usepackage[T1]{fontenc}
\usepackage{lmodern}
\usepackage{fullpage}
\usepackage{amsmath,amsthm,amssymb}
\usepackage{verbatim}
\usepackage{mathtools}
\usepackage{algorithm}
\usepackage{algorithmicx}
\usepackage{algpseudocode}
\usepackage{subfig}
\usepackage[toc,page]{appendix}
\usepackage[usenames,dvipsnames]{color}
\usepackage{caption}
\usepackage{thmtools}
\usepackage{thm-restate}
\usepackage{hyperref}
\usepackage{tabularx}
\usepackage[utf8]{inputenc} % allow utf-8 input
\usepackage[T1]{fontenc}    % use 8-bit T1 fonts
\usepackage{hyperref}       % hyperlinks
\usepackage{url}            % simple URL typesetting
\usepackage{booktabs}       % professional-quality tables
\usepackage{amsfonts}       % blackboard math symbols
\usepackage{nicefrac}       % compact symbols for 1/2, etc.
\usepackage{microtype}      % microtypography
\usepackage{enumitem}
\usepackage{textcomp}
\usepackage{comment}
\usepackage{color}

\usepackage[capitalize, nameinlink]{cleveref}
   
\makeatletter
\newcommand*{\rom}[1]{\expandafter\@slowromancap\romannumeral #1@}

\newcommand{\bethe}{\mathrm{bethe}}
\newcommand{\sink}{\mathrm{sinkhorn}}
\newcommand{\ssink}{\mathrm{scaledsinkhorn}}
\newcommand{\betheq}{\mathrm{F}}
\newcommand{\frst}{\mathrm{U}}
\newcommand{\scnd}{\mathrm{V}}
\newcommand{\mc}{\textbf{C}}
\newcommand{\my}{\textbf{Y}}
\newcommand{\vones}{\textbf{1}}

\newcommand{\prob}[1]{\mathrm{Pr}\left(#1 \right)}
\newcommand{\permD}{S_{\bX}}

\newcommand{\mj}{\textbf{m}_{j}}

\newcommand{\perm}{\mathrm{perm}}

\newcommand{\mx}{\textbf{X}}
\newcommand{\kds}{\R_{\geq 0}^{\bX \times \bX}}

\newcommand{\expt}[2]{\mathbb{E}_{#1}\left[#2\right]}
\newcommand{\KL}[1]{\mathrm{KL}\left(#1\right)}

\newcommand{\Axy}{\ma_{x,y}}

\newcommand{\Ns}{N}

\newcommand{\boo}{b_1}
\newcommand{\btt}{b_2}

\newcommand{\ztbtt}{[0,\btt]}
\newcommand{\otboo}{[1,\boo]}

\newcommand{\bttpo}{(\btt+1)}

\newcommand{\F}{\phi}

\newcommand{\1}{\overrightarrow{1}}
\newcommand{\R}{\mathbb{R}}

\newcommand{\bbP}{\mathbb{P}}
\newcommand{\E}{\mathbb{E}}
\newcommand{\Z}{\mathbb{Z}_{+}}

\newcommand{\simplex}{\Delta^{\bX}}
\newcommand{\psimplex}{\Delta_{pseudo}^{\bX}}
\newcommand{\dsimplex}{\Delta_{\bR}^{\bX}}

\newcommand{\probpml}{\bbP}

\newcommand{\bg}{\textbf{g}}

\newcommand{\bp}{\textbf{p}}
\newcommand{\bq}{\textbf{q}}
\newcommand{\bff}{\textbf{f}}

\newcommand{\bX}{\mathcal{D}}

\newcommand{\bQ}{\textbf{Q}}

\newcommand{\bZ}{\textbf{Z}}
\newcommand{\bR}{\textbf{R}}

\newcommand{\bS}{\textbf{S}}

\newcommand{\expo}[1]{\exp \left(#1 \right)}
\newcommand{\exps}[1]{\exp (#1 )}

\newcommand{\fnh}{\textbf{h}}

\newcommand{\otilde}{\widetilde{O}}

\newcommand{\lb}{\ell}
\newcommand{\ub}{u}

\newcommand{\cphi}{C_{\phi}}

\newcommand{\level}{\ell}

\newcommand{\ma}{\textbf{A}}
\newcommand{\mb}{\textbf{B}}

\newcommand{\pvec}{\zeta}

\newcommand{\eqdef}{\stackrel{\mathrm{def}}{=}}
\newcommand{\defeq}{\eqdef}

\newcommand{\onevec}{\overrightarrow{\mathrm{1}}}

\newcommand{\vecc}{c}

\newcommand{\poly}{\mathrm{poly}}

\newcommand{\setd}{\textbf{M}}
\newcommand{\eled}{\textbf{m}}
\newcommand{\dstoc}{\mathbf{K}_{rc}}

\newcommand{\disc}{\mathrm{disc}}
\newcommand{\mz}{\textbf{m}_{0}}

\newcommand{\bZfrac}{\textbf{Z}^{\phi,frac}_{\bR}}

\newcommand{\bqS}{\bq_{\bS}}
\newcommand{\ro}{\textbf{r}_{1}}

\newcommand{\rl}{\textbf{r}_{\ell}}
\newcommand{\ri}{\textbf{r}_{i}}
\newcommand{\zrm}{\R_{\geq 0}^{\ell \times [0,k]}}

\newcommand{\za}{\textbf{Z}_{\bR}^{{\bq},\phi}}

\newcommand{\lpi}{\level^{\bp}_{i}}
\newcommand{\Sij}{\bS_{i,j}}

\newcommand{\bigO}[1]{O\left(#1 \right)}
\newcommand{\bSp}{\bS'}

\newcommand{\logparam}{\Delta}
\newcommand{\bRext}{\bR^{\mathrm{ext}}}
\newcommand{\bZext}{\textbf{Z}^{\phi}_{\bRext}}
\newcommand{\bSext}{\textbf{S}^{\mathrm{ext}}}

\newcommand{\newk}{\bttpo}

\newcommand{\rones}{\textbf{1}}
\newcommand{\cones}{\textbf{1}}
\newcommand{\vr}{\textbf{r}}

\newcommand{\rmin}{\textbf{r}_{min}}

\DeclarePairedDelimiter{\abs}{\lvert}{\rvert}

\DeclarePairedDelimiter{\ceil}{\lceil}{\rceil}
\DeclarePairedDelimiter{\floor}{\lfloor}{\rfloor}

\DeclareMathOperator*{\argmax}{arg\,max}

\global\long\def\norm#1{\big\|#1\big\|}

\usepackage{thmtools}
\usepackage{thm-restate}
\usepackage{overpic}
\declaretheorem[name=Theorem,numberwithin=section]{thm}
\declaretheorem[name=Theorem,numberlike=thm]{theorem}
\declaretheorem[name=Lemma,numberlike=thm]{lemma}

\declaretheorem[name=Corollary,numberlike=thm]{cor}
\declaretheorem[name=Definition,numberlike=thm,style=definition]{defn}

\title{Instance Based Approximations to Profile Maximum Likelihood}

\author{%
Nima Anari\\
Stanford University\\
\texttt{anari@stanford.edu} \\
\and
Moses Charikar\\
Stanford University\\
\texttt{moses@cs.stanford.edu} \\
\and
Kirankumar Shiragur\\
Stanford University\\
\texttt{shiragur@stanford.edu} \\
\and
Aaron Sidford\\
Stanford University\\
\texttt{sidford@stanford.edu} \\
}

\begin{document}
	
\maketitle

\begin{abstract}
In this paper we provide a new efficient algorithm for approximately computing the profile maximum likelihood (PML) distribution, a prominent quantity in symmetric property estimation. We provide an algorithm which matches the previous best known efficient algorithms for computing approximate PML distributions and improves when the number of distinct observed frequencies in the given instance is small. We achieve this result by exploiting new sparsity structure in approximate PML distributions and providing a new matrix rounding algorithm, of independent interest. Leveraging this result, we obtain the first provable computationally efficient implementation of PseudoPML, a general framework for estimating a broad class of symmetric properties. Additionally, we obtain efficient PML-based estimators for distributions with small profile entropy, a natural instance-based complexity measure. Further, we provide a simpler and more practical PseudoPML implementation that matches the best-known theoretical guarantees of such an estimator and evaluate this method empirically. 
\end{abstract}
%!TEX root = main.tex

\section{Introduction}\label{sec:intro}
We consider the fundamental problem of \emph{symmetric property estimation}: given access to $n$ i.i.d.\ samples from an unknown distribution, estimate the value of a given symmetric property (i.e. one invariant to label permutation). This is an incredibly well-studied problem with numerous applications  \cite{Chao84, BF93, CCGLMCL12,TE87, Fur05, KLR99, PBGELLSD01, DS13, RCSWTKRWC09, GTPB07, HHRB01} and property-specific estimators, e.g. for support~\cite{VV11a, WY15}, support coverage~\cite{ZVVKCSLSDM16,OSW16}, entropy~\cite{VV11a, WY16, JVHW15}, and distance to uniformity~\cite{VV11b, JHW16}.

However, in a striking recent line of work it was shown that there is \emph{a universal approach} to achieving sample optimal\footnote{Sample optimality is up to constant factors. See \cite{ADOS16} for details.} estimators for a broad class of symmetric properties, including those above. \cite{ADOS16} showed that the value of the property on a distribution that (approximately) maximizes the likelihood of the observed profile (i.e. multiset of observed frequencies) is an optimal estimator up to accuracy\footnote{We use $\epsilon \gg n^{-c}$ to denote $\epsilon > n^{-c+\alpha}$ for any constant $\alpha>0$.} $\epsilon \gg n^{-1/4}$. Further, \cite{ACSS20} ,which in turn built on \cite{ADOS16,CSS19}, provided a polynomial time algorithm to compute an $ \exp(-O(\sqrt{n}\log n))$-approximate profile maximum likelihood distribution (PML).  Together, these results yield efficient sample optimal estimators for various symmetric properties up to accuracy $\epsilon \gg n^{-1/4}$.
  
Despite this seemingly complete picture of the complexity of PML, recent work has shown that there is value in obtaining improved approximate PML distributions. In \cite{CSS19pseudo,HO19} it was shown that variants of PML called \emph{PseudoPML} and \emph{truncated PML} respectively, which compute an approximate PML distribution on a subset of the coordinates, yield sample optimal estimators in broader error regime for a wide range of symmetric properties. Further, in \cite{HO20pentropy} an instance dependent quantity known as \emph{profile entropy} was shown to govern the accuracy achievable by PML and their analysis holds for all symmetric properties with no additional assumption on the structure of the property. Additionally, in \cite{HS20} it was shown that PML distributions yield a sample optimal universal estimator up to error $\epsilon \gg n^{-1/3}$ for a broad class of symmetric properties. However, the inability to obtain approximate PML distributions of approximation error better than $ \exp(-O(\sqrt{n}\log n))$ has limited the provably efficient implementation of these methods.
%Old Text
%In this paper we ask two fundamental questions regarding PML and its more recent practical applications. First, when can higher accuracy PML distributions be computed in polynomial time and second, when can these techniques be leveraged to yield more efficient (and possibly practical) approaches for estimating symmetric distributions? 

In this paper we enable many of these applications by providing improved efficient approximations to PML distributions. Our main theoretical contribution is a polynomial time algorithm that computes an $\exp(-O(k\log n))$-approximate PML distribution where $k$ is the number of distinct observed frequencies. As $k$ is always upper bounded by $\sqrt{n}$, our work generalizes the previous best known result from \cite{ACSS20} that computed an $ \exp(-O(\sqrt{n}\log n))$-approximate PML. Leveraging this result, our work provides the first provably efficient implementation of PseudoPML. Further, our work also yields the first provably efficient estimator for profile entropy and efficient estimators with instance-based high-accuracy guarantees via profile entropy. We obtain our approximate PML result by leveraging interesting sparsity structure in convex relaxations of PML~\cite{ACSS20,CSS19} and additionally provide a novel matrix rounding algorithm that we believe is of independent interest. %First, we show how to leverage homogeneity in the convex relaxation of PML and properties of basic feasible solutions of a linear program to efficiently obtain sparse approximate solutions to these convex relaxations. This reduces the problem of computing the desired PML approximation to a particular matrix rounding problem where we need to ``round down'' a matrix of nonnegative reals to another one with integral row and column sums without changing the entries too much in $\ell_1$. Perhaps surprisingly, we show that this is always possible by reduction to a combinatorial problem which we solve by combining seemingly disparate theorems from combinatorics and graph theory. Further, we show that this rounding can be computed efficiently by employing algorithms for enumerating near-minimum-cuts of a graph \cite{Karger96}.

Finally, beyond the above theoretical results we provide a simplified instantiation of these results that is sufficient for implementing PseudoPML. We believe this result is a key step towards practical PseudoPML. We provide preliminary experiments in which we perform entropy estimation using the PseudoPML approach implemented using our simpler rounding algorithm. Our results match other state-of-the-art estimators for entropy, some of which are property specific. 

\newcommand{\phik}{\psi}
\newcommand{\phikf}{h}
\newcommand{\Phik}{\Psi}
\newcommand{\Phis}{\Phi_{S}}
\newcommand{\Phikn}{\Psi^{(k,n)}}
\newcommand{\Fk}{\psi}
\newcommand{\dfreq}[1]{\mathrm{Freq}\left(#1\right)}
\newcommand{\fsubp}{\mathrm{F'}}
\newcommand{\consto}{2c_{1}}
\newcommand{\constt}{c_{5}}
\newcommand{\gset}{\mathrm{G}}
\newcommand{\bpmlp}{\bp_{\phi_{S'}}^{\beta}}
\newcommand{\phisp}{\phi_{S'}}
\newcommand{\phis}{\phi_{S}}
\newcommand{\rvS}{\textbf{S}}
\newcommand{\bgg}{\textbf{g}}
\newcommand{\Phisn}{\Phi_{S}^{n}}
\newcommand{\bpml}{\bp^{\beta}_{\phis}}
\newcommand{\fsub}{\mathrm{F}}
\newcommand{\Prob}[1]{\bbP\left(#1 \right)}
\newcommand{\xon}{x_1^n}
\newcommand{\xtn}{x_2^n}

\paragraph{Notation and basic definitions:}
Throughout this paper we assume we receive a sequence of $n$ independent samples from an underlying distribution $\bp \in \simplex$, where $\bX$ is a domain of elements and $\simplex$ is the set of all discrete distributions supported on this domain. We let  $[a,b]$ and $[a,b]_{\R}$ denote the interval of integers and reals $\geq a$ and $\leq b$ respectively, so $\simplex \defeq \{\bq \in [0,1]_{\R}^{\bX} | \norm{q}_1 = 1\}$.
%%We also use $[a]$ to denote interval $[1,a]$. 
%Let $\bX$ be the domain of elements and $\Ns\defeq |\bX|$.
%Let $\simplex \subset [0,1]_{\R}^{\bX}$ be the set of all discrete distributions supported on domain $\bX$. %and let $N$ be the size of the domain. 
%%\sidford{I think you can just use the word discrete when needed and use the simplex notation liberarlly - this is probably clearer in general and then you don't need this sentence.} \kiran{What should I call pseudo-distributions? and note we also have discrete pseudo-distributions.} 
%Throughout this paper we assume we receive a sequence of $n$ independent samples from an underlying distribution $\bp \in \simplex$. 

We let $\bX ^n$ be the set of all length $n$ sequences and $y^n \in \bX^n$ be one such sequence with $y^n_{i}$ denoting its $i$th element. We let $\bff(y^n,x)\defeq |\{i\in [n] ~ | ~ y^n_i = x\}|$ and $\bp_{x}$ be the frequency and probability of $x\in \bX$ respectively. For a sequence $y^n \in \bX^n$, let $\setd=\{ \bff(y^n,x) \}_{x \in \bX} \backslash \{0\}$ be the set of all its non-zero distinct frequencies and $\eled_1,\eled_2,\dots, \eled_{|\setd|}$ be these distinct frequencies. 

The \emph{profile} of a sequence $y^n$, denoted $\phi=\Phi(y^n)$, is a vector in $\Z^{|\setd|}$, where $\F_j \defeq |\{x\in \bX ~|~\bff(y^n,x)=\eled_{j} \}|$ is the number of domain elements with frequency $\eled_{j}$. We call $n$ the length of profile $\F$ and let $\Phi^n$ denote the set of all profiles of length $n$. The probability of observing sequence $y^n$ and profile $\phi$ with respect to a distribution $\bp$ are as follows,
$$\bbP(\bp,y^n) = \prod_{x \in \bX}\bp_x^{\bff(y^n,x)} \quad \text{ and } \quad \probpml(\bp,\phi)=\sum_{\{y^n \in \bX^n~|~ \Phi (y^n)=\phi \}} \bbP(\bp,y^n)~.$$
%We next formally define PML and approximate PML distributions.
%\begin{defn}
	For a profile $\phi \in \Phi^{n}$, $\bp_{\phi}$ is a \emph{profile maximum likelihood} (PML) distribution if $\bp_{\phi} \in \argmax_{\bp \in \simplex}$ $\probpml(\bp,\phi)$. 
	%and $\probpml(\bp_{pml,\phi},\phi)$ is the maximum PML objective value. 
	Further, a distribution $\bp^{\beta}_{\phi}$ is a $\beta$-\emph{approximate PML} distribution if $\probpml(\bp^{\beta}_{\phi},\phi)\geq \beta \cdot \probpml(\bp_{\phi},\phi)$.	
%\end{defn}

For a distribution $\bp$ and $n$, we let $\mx$ be a random variable that takes value $\phi \in \Phi^n$ with probability $\prob{\bp,\phi}$. The distribution of $\mx$ depends only on $\bp$ and $n$ and we call $H(\mx)$ (entropy of $\mx$) the \emph{profile entropy} with respect to $(\bp,n)$ and denote it by $H(\Phi^n,\bp)$.

We use $\otilde(\cdot)$, $\widetilde{\Omega}(\cdot)$ notation to hide all polylogarithmic factors in $n$ and $N$.

\paragraph{Paper organization:} In \Cref{sec:results} we formally state our results. In \Cref{sec:convex}, we provide the convex relaxation~\cite{CSS19,ACSS20} for the PML objective. Using this convex relaxation, in \Cref{sec:mainfirst} we state our algorithm that computes an $\exp(-O(k\log n))$-approximate PML and sketch its proof. Finally, in \Cref{sec:practical}, we provide a simpler algorithm that provably implements the PseudoPML approach; we implement this algorithm and provide experiments in the same section. Due to space constraints, we defer most of the proofs to appendix.

%!TEX root = main.tex

%\newpage
\section{Results}\label{sec:results}
Here we provide the main results of our paper. These include computing approximations to PML where the approximation quality depends on the number of distinct frequencies, as well as efficiently implementing results on profile entropy and PseudoPML.

\paragraph{Distinct frequencies:}
Our main approximate PML result is the following. 
\begin{theorem}[Approximate PML]\label{thm:resultmainone}
There is an algorithm that given a profile $\phi \in \Phi^n$ with $k$ distinct frequencies,  computes an $\expo{-O(k\log n)}$-approximate PML distribution in time polynomial in $n$.
\end{theorem}
%As the number of distinct observed frequencies in a profile of length $n$ is upper bounded by $\sqrt{n}$, our result is a generalization of Theorem 3.4 in \cite{ACSS20} that computes an $\exps{-O(\sqrt{n}\log n)}$-approximate PML. 
Our result generalizes \cite{ACSS20} which computes an $\exps{-O(\sqrt{n}\log n)}$-approximate PML. Through \cite{ADOS16} our result also provides efficient optimal estimators for class of symmetric properties when $\epsilon \gg n^{-1/4}$. Further, for distributions that with high probability output a profile with $O(n^{1/3})$ distinct frequencies, through \cite{HS20} our algorithm enables efficient optimal estimators for the same class of properties when $\epsilon \gg n^{-1/3}$. In \Cref{sec:mainfirst} we provide a proof sketch for the above theorem and defer the proof details to \Cref{app:approxpml}.
%Our theorem also has applications to the profile entropy and PseudoPML that we describe in greater detail.

%============================================================================
\paragraph{Profile entropy:}
One key application of our instance-based, i.e. distinct-frequency-based, approximation algorithm is the efficient implementation of the following approximate PML version of the profile entropy result from \cite{HO20pentropy}.\footnote{Theorem 3 in \cite{HO20pentropy} discuss instead exact PML and the authors discuss the approximate PML case in the comments; we confirmed the sufficiency of approximate PML claimed in the theorem through private communication with the authors.}. See \Cref{sec:intro} for the definition of profile entropy.
\begin{lemma}[Theorem 3 in \cite{HO20pentropy}]\label{lem:profileentropy}
	Let $f$ be a symmetric property. For any $\bp \in \simplex$ and a profile $\phi \sim \bp$ of length $n$ with $k$ distinct frequencies, with probability at least $1-O(1/\sqrt{n})$,
	$$|f(\bp)-f(\bp^{\beta}_{\phi})| \leq 2\epsilon_{f}\left(\frac{\widetilde{\Omega}(n)}{\ceil{H(\Phi^n,\bp)}}\right)~,$$
	where $\bp^{\beta}_{\phi}$ is any $\beta$-approximate PML distribution for $\beta>\exp(-O(k\log n))$ and $\epsilon_{f}(n)$ is
	the smallest error that can be achieved by any estimator with sample size $n$ and success proability at least $9/10$.\footnote{See \cite{HO20pentropy} for general success probability $1-\delta$; our work also holds for the general case.}
\end{lemma}
%\kiran{Shortened the comments, is this better?} 
%\moses{Do we need to make these selling point comments about the profile entropy result? Should we just refer readers to their paper instead? I'm not sure readers will necessarily be sold based on these brief comments.}
As the above result requires an $\exp(-O(k\log n))$-approximate PML, our \Cref{thm:resultmainone} immediately provides an efficient implementation of it. \Cref{lem:profileentropy} holds for any symmetric property with no additional assumptions on the structure. Further, it trivially implies a weaker result in \cite{ADOS16} where $\ceil{H(\Phi^n,\bp)}$ is replaced by $\sqrt{n}$. For further details and motivation, see \cite{HO20pentropy}.

% \kiran{Need to edit the following}
% The work of both \cite{ADOS16} and \cite{HS20} perform the worst case analysis of the performance of PML and they fail to derive condition on the error parameter that is more instance based and depends on the simplicity in the underlying distribution. 
 %The work of both \cite{ADOS16} and \cite{HS20} perform the worst case analysis of the performance of PML and they fail to derive condition on the error parameter that is more instance based and depends on the simplicity in the underlying distribution. 

%=============================================================================

\paragraph{PseudoPML:}
Our approximate PML algorithm also enables the efficient implementation of PseudoPML~\cite{CSS19pseudo,HO19}. Using PseudoPML, the authors in \cite{CSS19pseudo,HO19} provide a general estimation framework that is sample optimal for many properties in wider parameter regimes than the previous universal approaches. At a high level, in this framework, the samples are split into two parts based on the element frequencies. The empirical estimate is used for the first part and for the second part, they compute the estimate corresponding to approximate PML.
%Although the way to split the samples is dependent on the property (thus weakly depends on the property of interest), this approach gave a  
%Although computing the PML on a part of the sample (reducing the input size) made this approach practical, there are no known provable algorithms 
To efficiently implement the approach of PseudoPML required efficient algorithms with either strong or instance dependent approximation guarantees and our result (\Cref{thm:resultmainone}) achieves the later.
%\moses{This last sentence is confusing to me. Are you saying that reducing the input size does not yield an efficient pseudoPML implementation for all properties?}
%
%At a high level, the PseudoPML approach uses a mixture of PML and the empirical estimate to provide a general framework for symmetric property estimation. Although this approach weakly depends on the property, it provides estimators that are sample optimal in almost all interesting regimes.
We first state a lemma that relates the approximate PML computation to the PseudoPML.
%\begin{defn}[$\beta$-approximate proxy distribution]
%	Let $\bp \in \simplex$ be the hidden distribution. Given a sequence $x^n$ of i.i.d samples from distribution $\bp$ let $\phi \in \Phi^n$ be its corresponding profile. A distribution $\bq$ is an $\beta$-approximate proxy distribution if it satisfies the following,
%	$$\Prob{\bq,\phi} \geq \beta \cdot \Prob{\bp,\phi}~.$$
%\end{defn}

\begin{lemma}[PseudoPML]\label{lem:pseudomain}
	Let $\phi \in \Phi^n$ be a profile with $k$ distinct frequencies and $\ell,u \in [0,1]$.
	% and $\bp$ be a distribution where all its probability values lie in the interval $[\ell,u]$. 
	If there exists an algorithm that runs in time $T(n,k,u,\ell)$ and returns a distribution $\bp'$ such that
	\begin{equation}
	\label{eq:pseudopml_cond}
	\probpml(\bp',\phi) \geq \expo{-O((u-\ell)n\log n+k\log n)} \max_{\bq \in \simplex_{[\ell,u]}}\probpml(\bq,\phi)~,
	\end{equation}
	where $\simplex_{[\ell,u]}\defeq \{\bp\in \simplex\Big|\bp_x \in [\ell,u] ~\forall x\in \bX\}$. Then we can implement the PseudoPML approach with the following guarantees,
	\begin{itemize}[noitemsep,nosep,leftmargin=*]
	\item For entropy, when error parameter  $\epsilon>\Omega\left(\frac{\log N}{N^{1-\alpha}}\right)$ for any constant $\alpha>0$, the estimator is sample complexity optimal and runs in $T(n,O(\log n),O(\log n /n),1/\poly(n))$ time.
	\item For distance to uniformity, when $\epsilon>\Omega\left(\frac{1}{N^{1-\alpha}}\right)$ for any constant $\alpha>0$, the estimator is sample complexity optimal and runs in $T(n,\otilde(1/\epsilon),O(1/N),\Omega(1/N))$ time.
	\end{itemize}
	\end{lemma}
The proof of the lemma is divided into two main steps. In the first step, we relate \eqref{eq:pseudopml_cond} to conditions considered in PseudoPML literature.
%Further using this later condition, 
In the second step, we leverage this relationship and the analysis in \cite{CSS19pseudo,HO19} to obtain the result. See \Cref{app:pseduo} for the proof of the lemma and other details. As discussed in \cite{CSS19pseudo,HO19}, the above results are interesting because we have a general framework (PseudoPML approach) that is sample optimal in a broad range of non-trivial estimation settings; for instance when $\epsilon < \frac{\log N}{N}$ for entropy and $\epsilon <  \frac{1}{N^C}$ for distance to uniformity where $C>0$ is a constant, we know that the empirical estimate is optimal.

%We now describe the importance of the above result. For entropy, we already know from \cite{WY16} that the empirical distribution is sample complexity optimal if $\epsilon <c \frac{\log N}{N}$ for some constant $c>0$. Therefore the interesting regime for entropy estimation is when $\epsilon>\Omega\left(\frac{\log N}{N}\right)$ and our estimator works for almost all such $\epsilon$.
%Note that the estimator in \cite{JHW17} also requires that the error parameter $\epsilon \geq \frac{1}{N^C}$, where $C>0$ is some constant.

%\moses{Not sure I understand this next sentence. Don't we have $k \log n$ in the exponent in the PML approximation guarantee? As stated, the lemma needs $k$.}
As our approximate PML algorithm (\Cref{thm:resultmainone}) runs in time polynomial in $n$ (for all values of $k$) and returns a distribution that satisfies the condition of the above lemma; we immediately obtain an efficient implementation of the results in \Cref{lem:pseudomain}. However for practical purposes, we present a simpler and faster algorithm that outputs a distribution which suffices for the application of PseudoPML. 
%The distribution output by this algorithm may not necessarily be an approximate PML but it satisfies the condition of \Cref{lem:pseudomain}. 
We summarize this result in the following theorem.
%The algorithm presented in \Cref{thm:apml} is not practical and it turns out for the application of symmetric property estimation we do not require an approximate PML and all we need is an approximate proxy distribution defined below. Although the running time of our algorithm to compute an approximate PML is large, we have a practical algorithm to compute the approximate proxy distribution that we state next.
\begin{theorem}[Efficient PseudoPML]\label{thm:resultmaintwo}
	There exists an algorithm that implements \Cref{lem:pseudomain} in time $T(n,k,u,\ell)=\otilde(n~ k^{\omega-1} \log \frac{u}{\ell})$, where $\omega$ is the matrix multiplication constant. Consequently, this provides estimators for entropy and distance to uniformity in time $\otilde(n)$ and $\otilde({n/{\epsilon^{\omega-1}}})$ under their respective error parameter restrictions.
%	Let $\phi \in \Phi^n$ be a profile with $k$ distinct frequencies  and $\ell,u \in [0,1]$.
%	%and $\bp$ be any distribution where all its probability values lie in the interval $[\ell,u]$. 
%	There exists an algorithm that runs in time $\otilde(n~ \poly(k) \log \frac{u}{\ell})$ and returns a distribution $\bp'$ that satisfies,
%	$$\probpml(\bp',\phi) \geq \expo{-O((u-\ell)n+k)} \max_{\bq \in \simplex_{[\ell,u]}} \probpml(\bq,\phi)~.$$
%	Further this algorithm is used to implement the pseudo PML estimator with the following guarantees,
%	\begin{itemize}[noitemsep,nosep,leftmargin=*]
%		\item For entropy when error parameter $\epsilon>\Omega\left(\frac{\log N}{N^{1-\alpha}}\right)$ for any constant $\alpha>0$, the estimator is sample complexity optimal and runs in $\otilde(n)$ time.
%		\item For distance to uniformity when $\epsilon>\Omega\left(\frac{1}{N^{1-\alpha}}\right)$ for any constant $\alpha>0$, the estimator is sample complexity optimal and runs in $\otilde({n} ~\poly(1/{\epsilon}))$ time. \kiran{Verify the running time here.}
%	\end{itemize}
\end{theorem}
%Refer to \Cref{sec:practical} for the description of the algorithm and proof sketch. The running time mentioned in the above theorem is purely theoretical. Our algorithm needs to solve a convex relaxation for the PML objective and we use CVX~\cite{cvx} with package CVXQUAD~\cite{FSP17} to do it. We also use couple of heuristics to make our algorithm more practical and we discuss them in \Cref{app:experiments}.
See \Cref{sec:practical} for a description of the algorithm and proof sketch. The running time in the above result involves: solving a convex program, $n/k$ number of linear system solves of $k\times k$ matrices and other low order terms for the remaining steps. In our implementation we use CVX\cite{cvx} with package CVXQUAD\cite{FSP17} to solve the convex program. We use couple of heuristics to make our algorithm more practical and we discuss them in \Cref{app:experiments}.

\subsection{Related work}
PML was introduced by \cite{OSSVZ04}. Since then, many heuristic approaches \cite{OSSVZ04,ADMOP10,PJW17,Von12, Von14} have been proposed to compute an approximate PML distribution. Recent work of \cite{CSS19} gave the first provably efficient algorithm to compute a non-trivial approximate PML distribution and gave a polynomialy time algorithm to compute a $\exp(-O(n^{2/3}\log n))$ approximation. Their proof of this result is broadly divided into three steps.
In the first step, the authors in \cite{CSS19} provide a convex program that approximates the probability of a profile for a fixed distribution.
 In the second step, they perform minor modifications to this convex program to reformulate it as instead maximizing over all distributions while maintaining the convexity of the optimization problem. The feasible solutions to the modified convex program represent fractional distributions and in the third step, a rounding algorithm is applied to obtain a valid distribution. The approximation quality of this approach is governed by the first and last step and  \cite{CSS19} showed a loss of $\exp(-O(n^{2/3}\log n))$ for each and thereby obtained $\exp(-O(n^{2/3}\log n))$-approximate PML distribution. In follow up work, \cite{ACSS20} improved the analysis for the first step and then provided a better rounding algorithm in the third step to output an $\exp(-O(\sqrt{n}\log n))$-approximate PML distribution. The authors in \cite{ACSS20} showed that the convex program considered in the first step by \cite{CSS19} approximates the probability of a profile for a fixed distribution up to accuracy $\exp(-O(k\log n))$, where $k$ is the number of distinct observed frequencies in the profile. However they incurred a loss of $\exp(-O(\sqrt{n}\log n))$ in the rounding step; thus returning an $\exp(-O(\sqrt{n}\log n))$ PML distribution. To prove these results, \cite{CSS19} used a combinatorial view of the PML problem while \cite{ACSS20} analyzed the Bethe/Sinkhorn approximation to the permanent~\cite{Von12,Von14}.

Leveraging the connection between PML and symmetric property estimation, \cite{CSS19} and \cite{ACSS20} gave efficient optimal universal estimators for various symmetric properties when $\epsilon \gg n^{-1/6}$ and $\epsilon \gg n^{-1/4}$ respectively. The broad applicability of PML in property testing and to estimate other symmetric properties was later studied in \cite{HO19}. \cite{HS20} showed interesting continuity properties of PML distributions and proved their optimality for sorted $\ell_1$ distance and other symmetric properties when $\epsilon \gg n^{-1/3}$; no efficient version of this result is known yet.
%are known as it requires computation of $\expo{-n^{-1/3}\log n}$-approximate PML. 
%Further, \cite{CSS19pseudo,HO19} independently introduced and studied the performance of pseudo/truncated PML, a variant of PML.

%Valiant and Valiant~\cite{VV11a} adopted and rigorously analyzed a linear programming based approach for universal estimators proposed by \cite{ET76} and showed that it is sample complexity optimal in the constant error regime for estimating certain symmetric properties.
%(namely, entropy, support size, support coverage, and distance to uniformity). 
%Recent work of Han, Jiao and Weissman~\cite{HJW18} applied a local moment matching based approach in designing efficient universal symmetric property estimators for a single distribution. \cite{HJW18} achieves the optimal sample complexity in a broader error regimes for estimating the power sum function, support and entropy. 

There have been other approaches for designing universal estimators, e.g. \cite{VV11a} based on \cite{ET76}, \cite{HJW18} based on local moment matching, and variants of PML by \cite{CSS19pseudo,HO19} that weakly depend on the property. Optimal sample complexities for estimating many symmetric properties were also obtained by constructing property specific estimators, e.g.  
%support~\cite{VV11a, WY15}, support coverage~\cite{OSW16,ZVVKCSLSDM16}, entropy~\cite{VV11a, WY16, JVHW15},  distance to uniformity~\cite{VV11b, JHW16}, 
sorted $\ell_{1}$ distance \cite{VV11b, HJW18}, Renyi entropy~\cite{AOST14, AOST17}, KL divergence~\cite{BZLV16, HJW16} and others.

%
%\kiran{Include somewhere or remove:} first provided a concave approximation to the PML objective function using ideas that involved discretization, grouping sequences with same probability and focusing on a group with a largest probability. Later, the authors presented a rounding algorithm to re The approximation factor was later improved to $\exp(-O(\sqrt{n}\log n))$ by \cite{ACSS20}. 
%
%In \cite{CSS19}, a subset of the authors presented a convex relaxation for the PML objective with an approximation ratio lower bounded by $\exp(-O(n^{2/3}\log n))$;\sidford{I think a few word might be missing here} techniques involved discretization and grouping sequences with same probability. In follow up work, through improved analysis of the Sinkhorn approximation to permanent \sidford{citation?} and using the idea of probability discretization from \cite{CSS19}, the authors in \cite{ACSS20} showed that the same convex relaxation approximates the PML objective within a multiplicative factor of $\exp(-O(k\log n))$. However to compute an approximate PML distribution required rounding the solution of the convex relaxation to a valid distribution and the rounding algorithms provided in these prior works of \cite{CSS19} and \cite{ACSS20} had worst case approximation guarantees lower bounded by $\exp(-O(n^{2/3}\log n))$ and $\exp(-O(\sqrt{n}\log n))$ respectively.

\subsection{Overview of techniques}
Here we provide a brief overview of the proof to compute an $\exp(-O(k\log n))$-approximate PML distribution. As discussed in the related work, both \cite{CSS19,ACSS20} analyzed the same convex program; \cite{ACSS20} showed that this convex program approximates the probability of a profile for a fixed distribution up to a multiplicative factor of $\exp(-O(k\log n))$. However in the rounding step, their algorithms incurred a loss of $\exp(-O(n^{2/3}\log n))$ and $\exp(-O(\sqrt{n}\log n))$ respectively. Computing an improved $\exp(-O(k\log n))$-approximate PML distribution required a better rounding algorithm which in turn posed several challenges. We address these challenges by leveraging interesting sparsity structure in the convex relaxation of PML~\cite{ACSS20,CSS19} (\Cref{lem:sparse}) and provide a novel matrix rounding algorithm (\Cref{thm:matrixround}). 

In our rounding algorithm, we first leverage homogeneity in the convex relaxation of PML and properties of basic feasible solutions of a linear program to efficiently obtain a sparse approximate solution to the convex relaxation. This reduces the problem of computing the desired approximate PML distribution to a particular matrix rounding problem where we need to ``round down'' a matrix of non-negative reals to another one with integral row and column sums without changing the entries too much ($O(k)$ overall) in $\ell_1$. Perhaps surprisingly, we show that this is always possible by reduction to a combinatorial problem which we solve by combining seemingly disparate theorems from combinatorics and graph theory. Further, we show that this rounding can be computed efficiently by employing algorithms for enumerating near-minimum-cuts of a graph \cite{Karger96}.

\section{Convex Relaxation to PML}\label{sec:convex}
Here we define the convex program that approximates the PML objective. This convex program was initially introduced in \cite{CSS19} and analyzed rigorously in \cite{CSS19,ACSS20}. We first describe the notation and later state the theorem in \cite{ACSS20} that captures the guarantees of the convex program.

\textbf{Probability discretization:} 
Let $\bR \defeq \{\ri\}_{i \in [1,\ell]}$ be a finite discretization of the probability space, where $\ri =\frac{1}{2n^2}(1+\alpha)^{i}$ for all $i\in [1,\ell-1]$, $\rl=1$ and $\ell\defeq |\bR|$ be such that $\frac{1}{2n^2}(1+\alpha)^{\ell}>1$; therefore $\ell=O(\frac{\log n}{\alpha})$. Let $\vr \in \Z^{\ell}$ be a vector where the $i$'th element is equal to $\ri$.  We call $\bq \in [0,1]^{\bX}_{\R}$ a \emph{pseudo-distribution} if $\|\bq\|_1 \leq 1$ and a \emph{discrete pseudo-distribution} with respect to $\bR$ if all its entries are in $\bR$ as well. We use $\psimplex$ and $\dsimplex$ to denote the set of all pseudo-distributions and discrete pseudo-distributions with respect to $\bR$ respectively.  
%For any $\bq \in \dsimplex$, let $\vlq\in \Z^{\ell}$ be such that $\vlq_{i}=|\{y\in \bX~|~\bp_y=\ri \}|$ is equal to the number of domain elements with probability $\ri$. 
For all probability terms defined involving distributions $\bp$, we extend those definitions to pseudo distributions $\bq$ by replacing $\bp_{x}$ with $\bq_{x}$ everywhere. 
%For convenience we refer to $\probpml(\bq,\phi)$ for any pseudo-distribution $\bq$ as the ``probability'' of profile $\phi$ with respect to $\bq$. 
The effect of discretization is captured by the following lemma.

%If we let the probability discretization set $\bR=\{(1+\alpha)^{1-i}\}_{i \in [\ell]}$ for some $\alpha \in (0,1)$ and $\ell=O(\frac{\log n}{\alpha})$ to be the smallest integer such that $\frac{1}{4n^2} \leq (1+\alpha)^{1-\ell}\leq \frac{1}{2n^2}$, then for such an $\bR$ the following lemma holds.
\renewcommand{\bp}{\textbf{p}}
\begin{lemma}[Lemma 4.4 in \cite{CSS19}]\label{lem:probdisc}
	For any profile $\phi \in \Phi^{n}$ and distribution $\bp \in \simplex$, there exists $\bq\in \dsimplex$ that satisfies $\probpml(\bp,\phi) \geq \probpml(\bq,\phi) \geq \expo{-\alpha n-6}\probpml(\bp,\phi)$ and therefore, 
	$$\max_{\bp\in \simplex}\probpml(\bp,\phi) \geq \max_{\bq\in \dsimplex} \probpml(\bq,\phi) \geq \expo{-\alpha n-6}\max_{\bp\in \simplex}\probpml(\bp,\phi)~.$$
\end{lemma}

For any probability discretization set $\bR$, profile $\phi$ and $\bq \in \dsimplex$, we define the following sets that help lower and upper bound the PML objective by a convex program.
\renewcommand{\boo}{\ell}
\renewcommand{\btt}{k}
\renewcommand{\bttpo}{\btt}
\renewcommand{\ztbtt}{[\btt]}
\renewcommand{\pvec}{\textbf{r}}
\renewcommand{\bZ}{\textbf{Z}^{\phi}_{\bR}}
%\begin{align*}
%&\za \defeq \Big\{\bS \in \zrm~\Big|~\sum_{j \in [0,k]}\Sij=\lpi \text{ for all } i \in [1,\ell]   \text{ and }  \sum_{i \in [1,\ell]}\Sij=\phi_{j} \text{ for all }j \in [0,k] \Big\}~,\\
%&\bZ =\defeq \Big\{\bS \in \R_{\geq 0}^{\ell \times (k+1)} ~\Big|~  \sum_{i\in[1,\ell]}\bS_{i,j}=\phi_{j} \text { for all }j \in [1,k], \sum_{j\in[0,k]}\bS_{i,j} \in \Z \text { for all }i \in [1,\ell] \text{ and } \sum_{i \in [1,k]} \ri \sum_{j \in [0,k]}\bS_{i,j}\leq 1\Big\}~,\\
%&\bZfrac \defeq \Big\{\bS \in \zrm~\Big|~\sum_{i\in[1,\ell]}\bS_{i,j}=\phi_{j} \text { for all }j \in [1,k] \text{ and } \sum_{i \in [1,k]} \ri \sum_{j \in [0,k]}\bS_{i,j}\leq 1 \Big\}~.
%\end{align*}
\begin{align}
%&\za \defeq \Big\{\bS \in \zrm~\Big|~ \bS \rones=\vlq   \text{ and }  [\bS^\top \cones]_{j}=\phi_{j} \text{ for all }j \in [0,k] \Big\}~,\\
&\bZ \defeq \Big\{\bS \in \R_{\geq 0}^{\ell \times [0,k]} ~\Big|~ \bS \rones \in \Z^{\ell},  [\bS^\top \cones]_{j}=\phi_{j} \text { for all }j \in [1,k] \text{ and } \vr^{\top} \bS \rones\leq 1\Big\}~,\label{eq:zrphi} \\ 
&\bZfrac \defeq \Big\{\bS \in \R_{\geq 0}^{\ell \times [0,k]}~\Big|~[\bS^\top \cones]_{j}=\phi_{j} \text { for all }j \in [1,k] \text{ and } \vr^{\top} \bS \rones\leq 1 \Big\}~,\label{def:bzfrac}
\end{align}
where in the above definitions the $0$'th column corresponds to domain elements with frequency $0$ (unseen) and we use $\mz \defeq 0$. We next define the objective of the convex program.

Let $\mc_{ij}\defeq \mj \log \ri$ and for any $\bS \in \R_{\geq 0}^{\ell \times [0,k]}$ define,
\begin{equation}
\bg(\bS)\defeq \exp\Big(\sum_{i\in[1,\ell],j\in[0,k]}\left[\mc_{ij}\mx_{ij}-\mx_{ij}\log\mx_{ij}\right]+\sum_{i\in[1,\ell]}[\mx\vones]_{i}\log[\mx\vones]_{i}\Big)~.
\end{equation}
The function $\bg(\bS)$ approximates the $\probpml(\bq,\phi)$ term and the following theorem summarizes this result.
\begin{theorem}[Theorem 6.7 and Lemma 6.9 in \cite{ACSS20}]\label{pmlprob:approx}
	Let $\bR$ be a probability discretization set. Given a profile $\phi \in \Phi^n$ with $k$ distinct frequencies 
	%and discrete pseudo-distribution $\bq$ with respect to $\bR$. 
	the following inequalities hold,
	\begin{equation}
	\expo{-O(k \log n)}\cdot	\cphi \cdot	\max_{\bS \in \bZ}\bg(\bS) \leq \max_{\bq \in \dsimplex}\probpml(\bq,\phi) \leq \expo{\bigO{k \log n}} \cdot \cphi \cdot \max_{\bS \in \bZ}\bg(\bS)~,
	\end{equation}
	\begin{equation}
	\max_{\bq \in \dsimplex}\probpml(\bq,\phi) \leq \expo{\bigO{k \log n}} \cdot \cphi \cdot \max_{\bS \in \bZfrac}\bg(\bS)~,
	\end{equation}
	where $\cphi\defeq \frac{n!}{\prod_{j\in [1,k]}(\eled_{j}!)^{\phi_{j}}}$ is a term that only depends on the profile.
%	\begin{align*}
%	%&\expo{-O(k \log (\Ns+n))}\cdot	\cphi \cdot	\max_{\bS \in \za}\bg(\bS) \leq \probpml(\bp,\phi) \leq \expo{\bigO{k \log (\Ns/k)}} \cdot \cphi \cdot \max_{\bS \in \za}\bg(\bS)~,\\
%	&\\
%%	\end{equation}
%%\end{theorem}
%%
%%%The previous theorem provides an upper bound for the probability of profile with respect to any discrete pseudo-distribution. However one issue with this upper bound is that it is not efficiently computable because the set $\bZ$ is not a convex set (because of the integrality constraints). We relax these integrality constraints and define the following new set.
%%%\begin{defn}\label{def:bzfrac}
%%%	Let $\bZfrac\subseteq \R_{\geq 0}^{\ell \times (k+1)}$ be the set of non-negative matrices, such that any $\bS \in \bZfrac$ satisfies,
%%%	\begin{equation}
%%%	\sum_{i\in[1,\ell]}\bS_{i,j}=\phi_{j} \text { for all }j \in [1,k] \text{ and } \sum_{i \in [1,k]} \ri \sum_{j \in [0,k]}\bS_{i,j}\leq 1~.
%%%	\end{equation}
%%%\end{defn}
%%
%%\begin{theorem}\label{lem:dpmlapprox}
%%	Let $\bR$ be a probability discretization set. Given a profile $\phi$, the following holds,
%%	\begin{align*}
%	&\max_{\bq \in \dsimplex}\probpml(\bp,\phi) \leq \expo{\bigO{k \log n}} \cdot \cphi \cdot \max_{\bS \in \bZfrac}\bg(\bS)~.
%	\end{align*}
\end{theorem}
See \Cref{app:sparse} for citations related to convexity of the function $\bg(\bS)$ and running time to solve the convex program. For any $\bS \in \bZ$, define a pseudo-distribution associated with it as follows.
\begin{defn}\label{defn:distS}
	For any $\bS \in \bZ$, the discrete pseudo-distribution $\bqS$ associated with $\bS$ and $\bR$ is defined as follows: For any arbitrary $\sum_{j\in[0,k]}\bS_{i,j}$ number of domain elements assign probability $\ri$. Further $\bp_{\bS}\defeq \bqS/\|\bqS\|_1$ is the distribution associated with $\bS$ and $\bR$.
\end{defn}
Note that $\bqS$ is a valid pseudo-distribution because of the third condition in \Cref{eq:zrphi} and these pseudo distributions $\bp_{\bS}$ and $\bqS$ satisfy the following lemma. 
\begin{lemma}[Theorem 6.7 in \cite{ACSS20}]\label{lem:associateddist}
	Let $\bR$ and $\phi \in \Phi^n$ be any probability discretization set and a profile respectively. For any $\bS \in \bZ$, the discrete pseudo distribution $\bqS$ and distribution $\bp_{\bS}$ associated with $\bS$ and $\bR$ satisfies: $\expo{-O(k \log n)}\cphi \cdot \bg(\bS) \leq \probpml(\bq,\phi) \leq \probpml(\bp,\phi)~.$
%	\begin{equation}
%	\expo{-O(k \log n)}\cphi \cdot \bg(\bS) \leq \probpml(\bq,\phi) \leq \probpml(\bp,\phi)~.
%	\end{equation}
	\end{lemma}
%Note in the above lemma, the upper bound only depends on the profile \footnote{$\cphi$ has no dependency on $\phi_{0}$.} and we removed all dependencies related to distributions (and also $\phi_{0}$). 
%!TEX root = neurips_2020.tex
\section{Algorithm and Proof Sketch of \Cref{thm:resultmainone}}\label{sec:mainfirst}
\newcommand{\sparse}{\mathrm{Sparse}}
\newcommand{\mround}{\mathrm{MatrixRound}}
\newcommand{\create}{\mathrm{CreateNewProbabilityValues}}
Here we provide the algorithm to compute an $\expo{-O(k \log n)}$-approximate PML distribution, where $k$ is the number of distinct frequencies. We use the convex relaxation from \Cref{sec:convex}; the maximizer of this convex program is a matrix $\bS \in \bZfrac$ and its $i$'th row sum denotes the number of domain elements with probability $\ri$. As the row sums are not necessarily integral, 
%we have a fractional representation of the distribution. To return a true distribution, 
we wish to round $\bS$ to a new matrix $\bS'$ that has integral row sums and $\bS' \in \textbf{Z}^{\phi}_{\bR'}$ for some probability discretization set $\bR'$. Our algorithm does this rounding and incurs only a loss of $\expo{-O(k \log n)}$ in the objective; finally the distribution associated with $\bS'$ and $\bR'$ is the desired $\expo{-O(k \log n)}$-approximate PML. 
%During the process of our rounding we prove interesting properties of the approximate PML and provide a novel matrix rounding algorithm.
We first provide a general algorithm that holds for any probability discretization set $\bR$ and the guarantees of this algorithm are stated below.
\begin{theorem}\label{thm:generalone}
	%There is an algorithm that given a profile $\phi \in \Phi^n$ with $k$ distinct frequencies and $\bR$,  computes an $\expo{-O(k\log n)}$-approximate PML distribution in time polynomial in $n$.
	Given a profile $\phi \in \Phi^n$ with $k$ distinct observed frequencies and $\bR$, there exists an algorithm that runs in polynomial of $n$ and $|\bR|$ time and returns a distribution $\bp'$ that satisfies,
	$$\Prob{\bp', \phi} \geq \expo{-O(k \log n)} \max_{\bq \in \dsimplex}\Prob{\bq, \phi} ~.$$
\end{theorem}
For an appropriately chosen $\bR$, the above theorem immediately proves \Cref{thm:resultmainone} and we defer its proof to \Cref{app:thmgeneralone}. In the remainder of this section we focus our attention towards the proof of \Cref{thm:generalone} and we next provide the algorithm that satisfies the guarantees of this theorem.
\begin{algorithm}[H]
	\caption{ApproximatePML$(\phi,\bR)$}\label{euclid}
	\begin{algorithmic}[1]
		%\Procedure{ApproximatePML}{$\phi,\bR$}
		%\State {\bf Input}: Profile $\phi \in \Phi^n$ and probability discretization set $\bR$.
		%\State {\bf Output}: A distribution $\bpapprox$.
		\State Solve $\bSp=\argmax_{\bS\in \bZfrac} \log \bg(\bS)$.  \Comment{Step 1}
		\State $\bS'' = \sparse(\bSp)$. \Comment{Step 2}
		\State $(\bS'',\mb'')=\mround(\bS'')$. \Comment{Step 3}
		\State $(\bSext,\bRext)=\create(\bS'',\mb'',\bR)$.  \Comment{Step 4}
		\State Return distribution $\bp'$ with respect to $\bSext$ and $\bRext$ (See \Cref{defn:distS}). \Comment{Step 5}
		%by assigning $(\bS \onevec)_{i}$ number of domain elements probability $\ri$ 
	%	\State \textbf{return} $\bpapprox\defeq \frac{\qext}{\|\qext\|_1}$.
		%\EndProcedure
	\end{algorithmic}
	\label{alg:final_temp}
\end{algorithm}
We divide the analysis of the above algorithm into 5 main steps. See \Cref{lem:associateddist} for the guarantees of Step 5 and here we state results for the remaining steps; we later combine it all to prove \Cref{thm:generalone}.
\begin{lemma}[\cite{CSS19,ACSS20}]\label{lem:maximizer}
	Step 1 of the algorithm can be implemented in $\otilde(|\bR|~k^2)$ time and the maximizer $\bSp$ satisfies: 
	$\cphi \cdot \bg(\bSp)\geq \expo{\bigO{-k \log n}} \max_{\bq \in \dsimplex}\probpml(\bq,\phi)$.
\end{lemma}
The running time follows from Theorem 4.17 in \cite{CSS19} and the guarantee of the maximizer follows from Lemma 6.9 in \cite{ACSS20}.
The lemma statements for the remaining steps are written in a general setting; we later invoke each of these lemmas in the context of the algorithm to prove \Cref{thm:generalone}.
\begin{lemma}[Sparse solution]\label{lem:sparse}
	For any $\ma \in \bZfrac$, the algorithm $\sparse(\ma)$ runs in $\otilde(|\bR|~k^{\omega})$ time and returns a solution $\ma' \in \bZfrac$ such that $\bg(\ma') \geq \bg(\ma)$ and $\big|\{i \in [1,\ell]~|~[\ma'\1]_{i}  > 0\}\big| \leq k+1$.
\end{lemma}
We defer description of the algorithm $\sparse(\mx)$ and the proof to \Cref{app:sparse}. In the proof, we use homogeneity of the convex program to write an LP whose optimal basic feasible solution satisfies the lemma conditions.
\begin{theorem}\label{thm:matrixround}
	For a matrix $\ma\in \R_{\geq 0}^{s\times t}$, the algorithm $\mround(\ma)$ runs in time polynomial in $s,t$ and returns a matrix $\mb\in \R_{\geq 0}^{s\times t}$ such that $\mb_{ij} \leq \ma_{ij}~\forall ~i \in [s],j \in [t]$, $\mb\1\in \Z^s$, $\mb^\top \1\in \Z^t$ and $\sum_{i,j} (\ma_{ij}-\mb_{ij})\leq O(s'+t')$, where $s',t'$ denote the number of non-zeros rows and columns.
\end{theorem}
For continuity of reading, we defer the description of $\mround(\ma)$ and its proof to \Cref{subsec:matrixround}.
\begin{lemma}[Lemma 6.13 in \cite{ACSS20}]\label{lem:create}
	For any $\ma\in \bZfrac \subseteq \R_{\geq 0}^{\ell \times [0,k]}$ and $\mb \in \R_{\geq 0}^{\ell \times [0,k]}$ such that $\mb_{ij} \leq \ma_{ij}$ for all $i \in [\ell],j \in [0,k]$, $\mb\1\in \Z^\ell$, $\mb^\top \1\in \Z^{[0,k]}$ and $\sum_{i\in [\ell], j \in [0,k]} (\ma_{ij}-\mb_{ij})\leq t$. The algorithm $\create(\ma,\mb,\bR)$ runs in polynomial time and returns a solution $\ma'$ and a probability discretization set $\bR'$ such that $\ma' \in \textbf{Z}^{\phi}_{\bR'}$ and $\bg(\ma') \geq \expo{-O\left(t\log n\right)}\bg(\ma)~.$
	%$$\bg(\ma') \geq \expo{-O\left(k\log\frac{k}{\textbf{r}_{\min}}\right)}\bg(\ma)~.$$
	%$$\bg(\ma') \geq \expo{-O\left(t\log n\right)}\bg(\ma)~.$$
\end{lemma}
The algorithm $\create$ is the same algorithm from \cite{ACSS20} and the above lemma is a simplified version of Lemma 6.13 in \cite{ACSS20}; see \Cref{app:create} for its proof. 

%Combining all the above results we now prove \Cref{thm:generalone}.
The proof of \Cref{thm:generalone} follows by combining results for each step and we defer it to \Cref{app:thmgeneralone}. 

\newcommand{\set}[1]{\{#1\}}
\newcommand{\card}[1]{|#1|}
\subsection{Matrix rounding algorithm and proof sketch of \Cref{thm:matrixround}}\label{subsec:matrixround}
In this section we prove \cref{thm:matrixround}. Given a matrix $\ma\in \R^{s\times t}_{\geq 0}$, our goal is to produce a rounded-down matrix $\mb$ with integer row and column sums, such that $0\leq \mb\leq \ma$ (entry wise) and the total amount of rounding $\sum_{i, j}(\ma_{ij}-\mb_{ij})$ is bounded by $O(s'+t')$, where $s', t'$ are the number of nonzero rows and columns respectively. For simplicity we may assume $s=s'$ and $t=t'$ by simply dropping the zero rows and columns from $\ma$ and re-appending them to the resulting $\mb$. As our first step, we reduce the problem to a statement about graphs. Below we use $\deg_F(v)$ to denote the number of edges adjacent to a vertex $v$ within a set of edges $F$.
%\kiran{I slightly tweaked the Theorem statement, can you have a look at it. We can just add a sentence or two saying why we can focus only on the non-zeros rows and columns.}
\begin{lemma}\label{lem:mod}
	Suppose that $G=(V, E)$ is a bipartite graph and $k$ is a positive integer. There exists a polynomial time algorithm that outputs a subgraph $F\subseteq E$, such that $\deg_F(v)=0$ modulo $k$ for every vertex $v$, and $\card{E-F}\leq O(k\card{V})$.
\end{lemma}
\begin{proof}[Proof of \cref{lem:mod} $\implies$ \cref{thm:matrixround}]
	Let $k=\min(s, t)$. Given $\ma$ we produce a bipartite graph with $s$ and $t$ vertices on two sides; for every entry $\ma_{ij}$ we round down to the nearest integer multiple of $1/k$, say $c_{ij}/k$, and introduce $c_{ij}$ parallel edges between vertices $i$ and $j$ of the bipartite graph. Now \cref{lem:mod} produces a subgraph $F$, and we let $\mb_{ij}$ be $1/k$ times the number of edges left in $F$ between $i,j$. By \cref{lem:mod}, $\mb$ will have integer row and column sums, and $0\leq \mb\leq \ma$. We next show that the total amount of rounding is bounded by $O(s+t)$.
	
	Notice that when rounding each entry of $\ma$ down to $c_{ij}/k$, the total amount of change is at most $st/k=O(s+t)$. By the guarantee that $\card{E-F}\leq O(k\card{V})$, the total amount of rounding in the second step is also bounded by $O(k(s+t))/k=O(s+t)$.
\end{proof}

So it remains to prove \cref{lem:mod}. As our main tool, we will use a result from \cite{thomassen2014graph} which was obtained by reduction to an earlier result from \cite{LMTWZ13}. Roughly, this result says that as long as $G$ is sufficiently connected, we can choose a subgraph whose degrees are \emph{arbitrary} values modulo $k$.
\begin{lemma}[{\cite[Theorem 1]{thomassen2014graph}}]\label{lem:connected}
	Suppose that $G=(V, E)$ is a bipartite $(3k-3)$-edge-connected graph. Suppose that $f:V \to \set{0, \dots, k-1}$ is an arbitrary function, with the restriction that the sum of $f$ on either side of the bipartite graph $G$ yields the same result modulo $k$. Then, there is a subgraph $F\subseteq E$, such that for each vertex $v$, $\deg_F(v)=f(v)$ modulo $k$.
\end{lemma}
Note that $(3k-3)$-edge-connectivity means that for every cut, i.e., every partitioning of vertices into two nonempty sets $S, S^c$, the number of edges between $S$ and $S^c$ is $\geq 3k-3$. We show that \cref{lem:connected} can also be made constructive, giving the polynomial time guarantee for \cref{lem:mod}.
\begin{lemma}\label{lem:connected-alg}
	There is a polynomial time algorithm that produces the subgraph of \cref{lem:connected}.	
\end{lemma}
We defer the proof of \cref{lem:connected-alg} to \cref{app:matrixround}. At a high level, the proof of \cref{lem:connected} works by formulating an assumption about the graph that is more general and more nuanced than edge-connectivity; instead of a constant lower bound on every cut, this assumption puts a cut-specific lower bound on each cut, the details of which can be found in \cref{app:matrixround}. The rest of the argument follows a clever induction. To make this argument constructive, we show how to check the nuanced variant of edge-connectivity in polynomial time. We do this by proving that only cuts of size smaller than a constant multiple of the minimum cut have to checked, and these can be enumerated in polynomial time \cite{Karger96}.

Note that \cref{lem:connected} does not guarantee anything about $\card{E-F}$, even when $f$ is the zero function (the empty subgraph is actually a valid answer in that case). We will fix this using a theorem of \cite{nash1961edge}.  We will first prove \cref{lem:mod} with the extra assumption that $G$ is $6k$-edge-connected, and then prove the general case.
\begin{proof}[Proof of \cref{lem:mod} when $G$ is $6k$-edge-connected]
	By a famous theorem due to \cite{nash1961edge}, a $6k$-edge-connected graph contains $6k/2=3k$ edge-disjoint spanning trees. Moreover the union of these $3k$ edge-disjoint spanning trees can be found in polynomial time by matroid partitioning algorithms \cite{GW92}. Let $H$ be the subgraph formed by these $3k$ edge-disjoint spanning trees. We will ensure that all edges outside $H$ are included in $F$; as a consequence, we will automatically get that $\card{E-F}$ is bounded by the number of edges in $H$, which is at most $3k(\card{V}-1)=O(k\card{V})$.
	
	Let $H^c$ denote the complement of $H$ in $G$. Define the function $f:V\to\set{0, \dots, k-1}$ in the following way: let $f(v)$ be $-\deg_{H^c}(v)$ modulo $k$. Note that $f$ has the same sum on either side of the bipartite graph, modulo $k$. We will apply \cref{lem:connected,lem:connected-alg} to the graph $H$ (which is $3k\geq (3k-3)$-edge-connected) and the function $f$. Then we take the union of the subgraph returned by \cref{lem:connected-alg} and $H^c$ and output the result as $F$. Then $\deg_F(v)=\deg_{H^c}(v)+f(v)=0$ modulo $k$, for every vertex $v$. Note again that since we only deleted edges in $H$ to get $F$, the total number of edges we have removed can be at most $O(k\card{V})$.
\end{proof}
%Now that we have proven \cref{lem:mod} for highly-connected graphs, we prove it for the general case.
We have shown \cref{lem:mod} for highly-connected graphs and the proof for the general case follows by partitioning the graph into union of vertex-disjoint highly-connected subgraphs while removing a small number of edges. We defer the proof for this general case to \cref{app:matrixround}.
%\input{rounding.tex}
%\input{final_rounding.tex}
%\input{missingproofs.tex}
%\input{symmetric.tex}
%!TEX root = neurips_2020.tex
\renewcommand{\lb}{\ell}
\renewcommand{\ub}{u}
\newcommand{\rmax}{\textbf{r}_{\max}}
\section{Algorithm, Proof Sketch of \Cref{thm:resultmaintwo} and Experiments}\label{sec:practical}
Here we present a simpler rounding algorithm that further provides a faster implementation of the pseudo PML approach with provable guarantees. Similar to \Cref{sec:mainfirst}, we first provide an algorithm with respect to a probability discretization set $\bR$ that proves \Cref{thm:secmain}; we later choose the discretization set carefully to prove \Cref{thm:resultmaintwo}. 
%As discussed in \Cref{sec:results}, \Cref{thm:resultmaintwo} helps implement the PseudoPML approach and 
We perform experiments in \Cref{subsec:expirements} to analyze the performance of this rounding algorithm empirically. We defer all remaining details to \Cref{app:pseudoall}.
\begin{theorem}\label{thm:secmain}
	Given a probability discretization set $\bR$ ($\ell\defeq |\bR|$) and a profile $\phi \in \Phi^n$ with $k$ distinct frequencies, there is an algorithm that runs in time $\otilde(\ell k^{\omega})$ and returns a distribution $\bp'$ such that,
	$$\Prob{\bp', \phi} \geq \expo{-O((\rmax-\rmin)n +k \log (\ell n))} \max_{\bq \in \dsimplex}\Prob{\bq, \phi} ~.$$
\end{theorem}
For an appropriately chosen $\bR$, the above theorem immediately proves \Cref{thm:resultmaintwo} and we defer both their proofs to \Cref{app:pseudoapproxpml}. 
We now present the algorithm that proves \Cref{thm:secmain}. 
%and to understand this algorithm recall that the elements of $\bR$ satisfy $\ro<\dots<\ri<\dots<\rl$.
%For an appropriately chosen $\bR$, the above theorem immediately proves \Cref{thm:resultmaintwo} and we defer both their proofs to \Cref{app:pseudoapproxpml}. 
%We now provide description of the algorithm that proves \Cref{thm:secmain} and to understand this algorithm recall that the elements of $\bR$ satisfy $\ro<\dots<\ri<\dots<\rl$.
%We also defer the description of the pseudo PML approach to \Cref{app:pseudoround} and just provide the experimental results (See \Cref{subsec:expirements}) using this approach.
%The description of our rounding algorithm that proves \Cref{thm:secmain} is stated below. 
\begin{algorithm}[H]
	\caption{ApproximatePML2$(\phi,\bR)$}\label{alg:round2}
	\begin{algorithmic}[1]
		%\Procedure{Rounding Algorithm 2}{$\phi,\bR,\bS'$}
		%\State {\bf Input}: Profile $\phi \in \Phi^n$, probability discretization set $\bR$ and a solution $\bS' \in \bZfrac$.
		%\State {\bf Output}: A distribution $\bp'$.
		\State Solve $\mx=\argmax_{\bS\in \bZfrac} \log \bg(\bS)$ and let $\mx'=\sparse(\mx)$.  \Comment{Step 1}
		%\State $\bS''=\sparse(\bS')$. \Comment{Step 2}
		\State Let $\bS'$ be the sub matrix of $\mx'$ corresponding to its non-zero rows.  \Comment{Step 2}
		\State Let $\bR'$ denote the elements in $\bR$ corresponding to non-zero rows of $\mx'$. Let $\ell'\defeq |\bR'|$. \Comment{Step 3}
		\For{$i=1\dots \ell'-1$} \Comment{Step 4}
		\State $\bSext_{i,j}=\bS'_{i,j}\frac{\floor{\|\bS'_i\|_1}}{\|\bS'_i\|_1}$ for all $j \in [0,k]$. \Comment{Step 5}
		\State $\bS'_{i+1,j}=\bS'_{i+1,j} + (\bS'_{i,j}-\bSext_{i,j})$ for all $j \in [0,k]$. \Comment{Step 6}
		\EndFor \Comment{Step 7}
		\State $\bSext_{\ell',j}=\bS'_{{\ell'},j}\frac{\floor{\|\bS'_{\ell'}\|_1}}{\|\bS'_{\ell'} \|_1}$ for all $j \in [0,k]$. \Comment{Step 8}
		\State Let $c=\sum_{i\in [1,\ell']}\ri' \|\bSext_{i}\|_{1}$, where $\ri'$ are the elements of $\bR'$. \Comment{Step 9}
		\State Define $\bRext=\{\ri''\}_{i\in [1,\ell']}$, where $\ri''=\frac{\ri'}{c}$ for all $i \in [1,\ell']$. \Comment{Step 10}
		\State Return distribution $\bp'$ with respect to $\bSext$ and $\bRext$ (See \Cref{defn:distS}).\Comment{Step 11}
		%by assigning $(\bS \onevec)_{i}$ number of domain elements probability $\ri$ 
		%\State \textbf{return} $\bp'$.
		%\EndProcedure
	\end{algorithmic}
	\label{alg:final}
\end{algorithm}

%================================================================
%
%\kiran{Will remove this part and push to appendix, ignore it for now.}
%
%The \Cref{lem:maximizer} already provides the guarantee of the Step 1 of the algorithm. We now provide the guarantees of the Steps 2-6.
%\begin{lemma}\label{lem:inter}
%	For any  $\mx \in \bZfrac$, the Steps 2-6 run in $O(k \ell)$ and return a probability discretization set $\bR'$ and $\mx' \in \R'$ that satisfy $\mx' \in \textbf{Z}^{\phi}_{\bR'}$ and 
%	$$\bg(\mx'') \geq \exp(-(\rmin-\rmax)n-k) \bg(\mx)$$~.
%\end{lemma}
%\begin{proof}
%	We defer the proof of this lemma to \Cref{app:pseudoround}
%\end{proof}
%
%\begin{proof}[Proof of \Cref{thm:secmain}]
%The distribution $\bp'$ is the discrete pseudo distribution corresponding to $\mx'$ and $\bR'$ satisfies,
%$$\Prob{\bp'', \phi} \geq \exp(-(\rmin-\rmax)n-k) \bg(\mx)~.$$
%Talk about in the proof step 7.
%\end{proof}
%
%================================================================
\subsection{Experiments}\label{subsec:expirements}
Here we present experimental results for entropy estimation. We analyze the performance of the PseudoPML approach implemented using our rounding algorithm with the other state-of-the-art estimators. 
Each plot depicts the performance of various algorithms for estimating entropy of different distributions with domain size $N=10^5$. The x-axis corresponds to the sample size (in logarithmic scale) and the y-axis denotes the root mean square error (RMSE). 
%The horizontal axis is the sample size and the vertical is the root mean square error. 
Each data point represents 50 random trials. %``Uniform'' is the uniform distribution, 
``Mix 2 Uniforms'' is a mixture of two uniform distributions, with half the probability mass on
the first $N/10$ symbols and the remaining mass on the last $9N/10$ symbols, and $\mathrm{Zipf}(\alpha) \sim 1/i^{\alpha}$ with $i \in [N]$. 
MLE is the naive approach of using empirical distribution with correction bias;
%, PJW is \cite{PJW17}, VV is \cite{VV11a}, and JVHW is \cite{JVHW15}
all the remaining algorithms are denoted using bibliographic citations. 

%Here we present experimental results for entropy estimation. We analyze the performance of the PseudoPML approach implemented using our rounding algorithm with the other state-of-the-art estimators. 
%Each plot depicts the performance of various algorithms for estimating entropy of different distributions with domain size $N=10^5$. The x-axis corresponds to the sample size (in logarithmic scale) and the y-axis denotes the root mean square error (RMSE). 
%Each data point represents 50 random trials. ``Mix 2 Uniforms'' is a mixture of two uniform distributions, with half the probability mass on
%the first $N/10$ symbols, and $\mathrm{Zipf}(\alpha) \sim 1/i^{\alpha}$ with $i \in [N]$. 
%MLE is the naive approach of using empirical distribution with correction bias;
%all the remaining algorithms are denoted using bibliographic citations. 

\begin{figure}[ht]
	\begin{center}
		\begin{tabular}{ccc}
			\begin{overpic}[width=.3\columnwidth]{%,grid]{%
					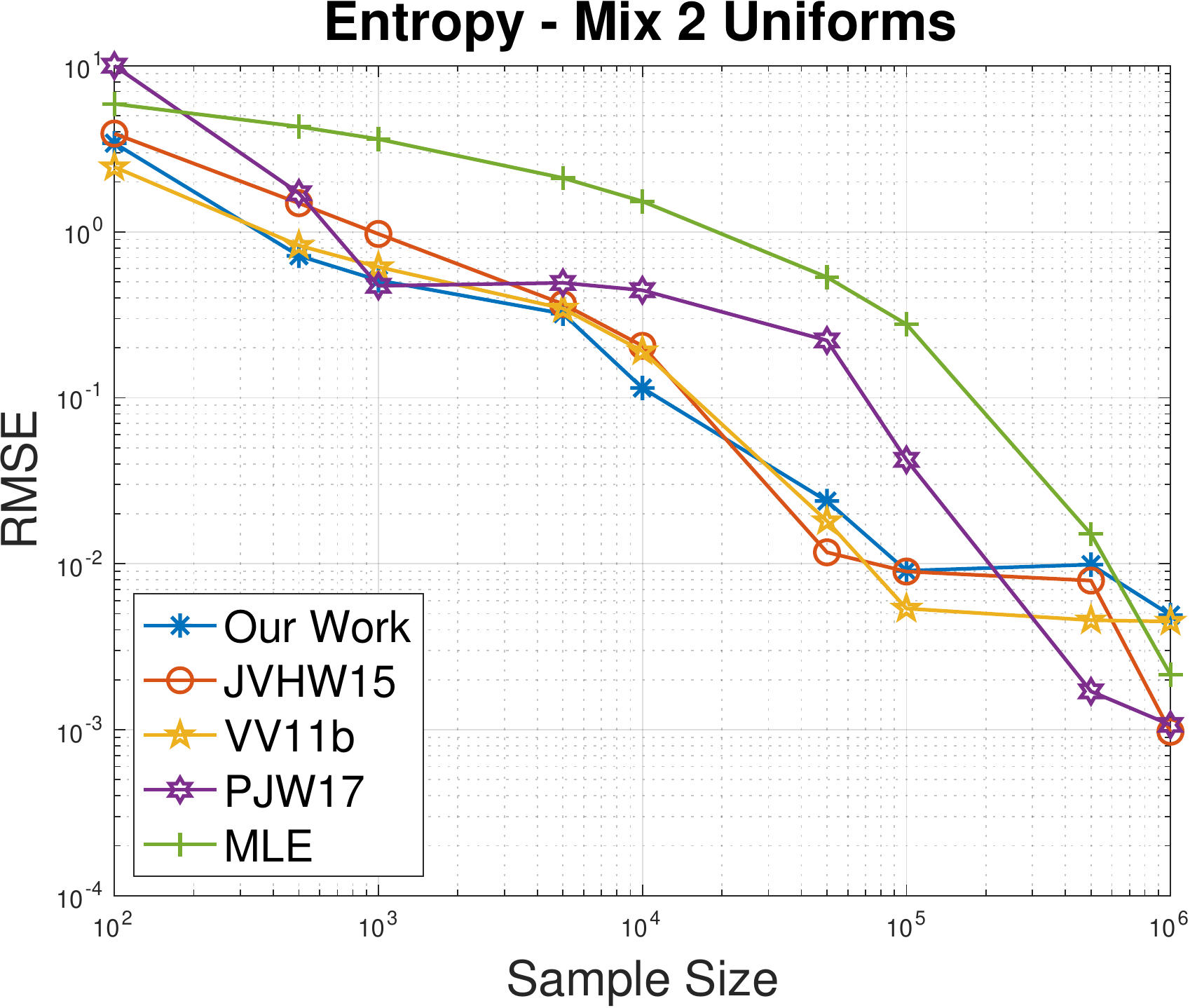}
			\end{overpic} &
			\begin{overpic}[width=.3\columnwidth]{%,grid]{%
					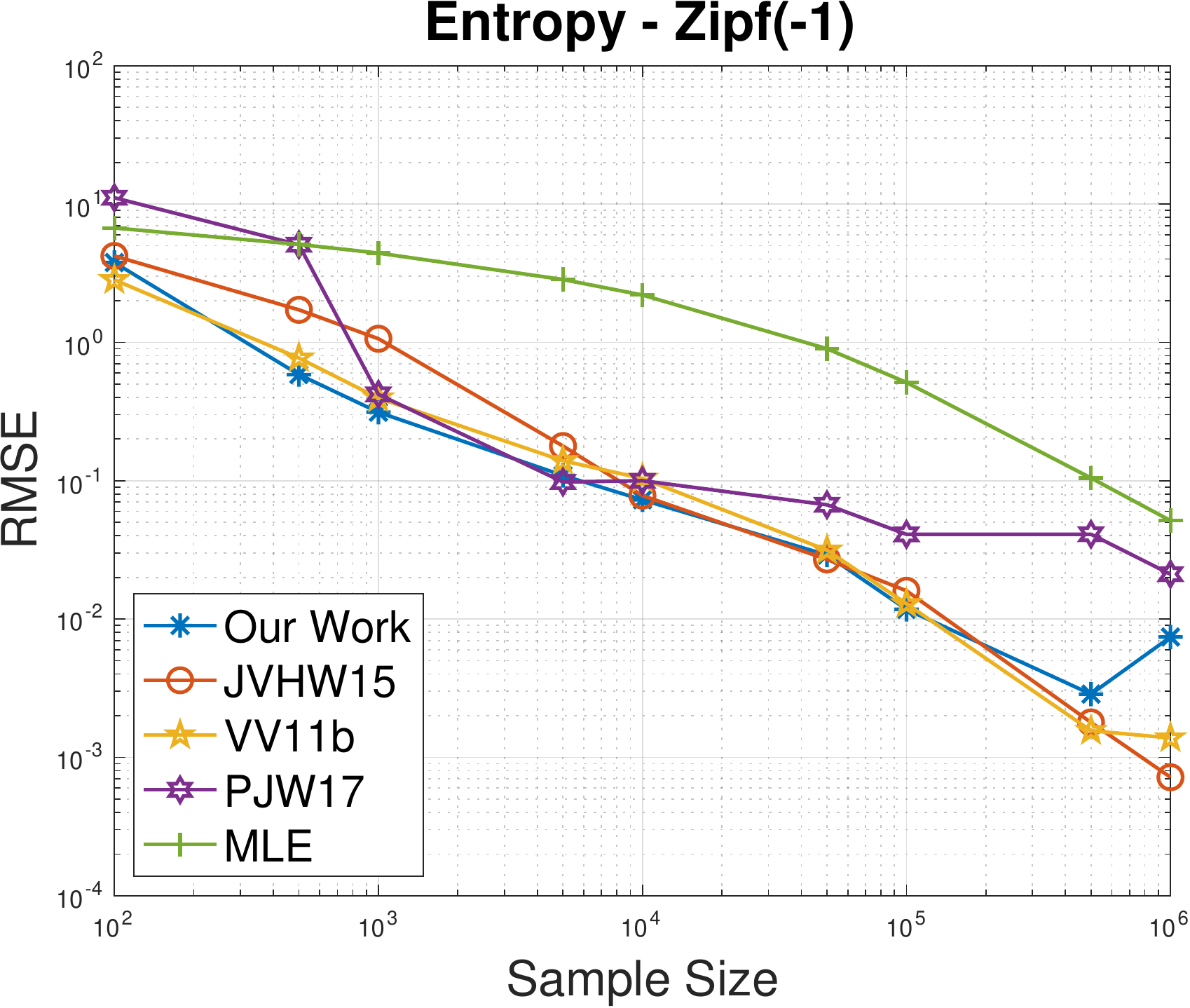}
			\end{overpic} &
			\begin{overpic}[width=.3\columnwidth]{%,grid]{%
					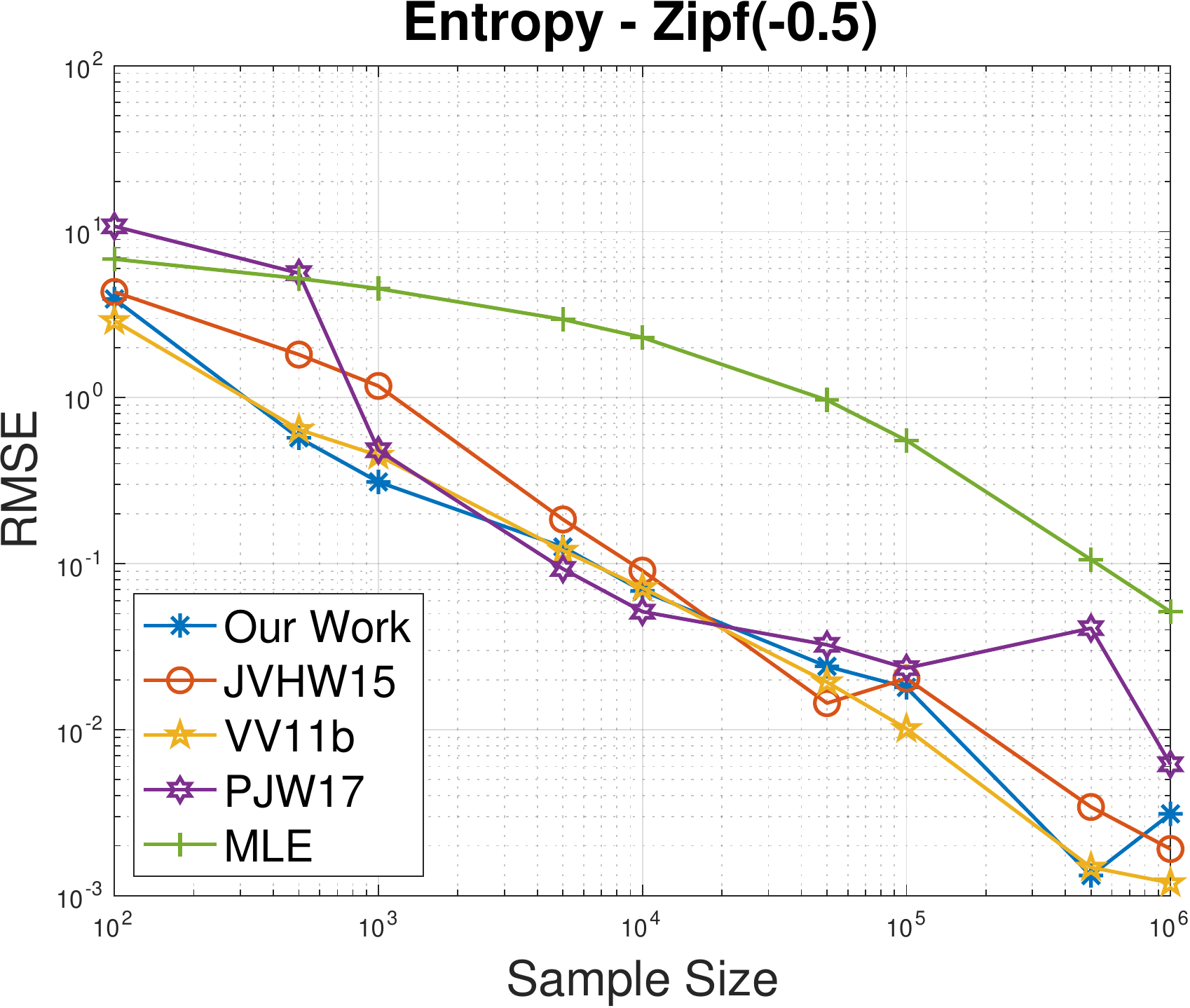}
			\end{overpic}
%			(a) & (b) & (c)
		\end{tabular}
%		\caption{
%			\label{fig:methods-intuition}
%			Experimental results.
%		}
\caption{
	\label{fig1}
	Experimental results for entropy estimation.
}
	\end{center}
\end{figure}

In the above experiment, 
%the other estimators that we compare to fall in the following different categories: heuristic \cite{PJW17}, property specific \cite{JVHW15} and universal approach for constant error regime \cite{VV11b}. For estimating entropy, 
note that the error achieved by our estimator is competitive with the other state-of-the-art estimators. As for the running times in practice, the other approaches tend to perform better than the current implementation of our algorithm. 
%It is not surprising that the algorithms in the first three categories seem to have better run times than provable PML based approaches as they are either fine tuned for the specific property, lack statistical guarantees, or provably work in a limited error regime. 
To further improve the running time of our approach or any other provable PML based approaches involves building an efficient practical solver for the convex optimization problem~\cite{CSS19,ACSS20} stated in the first step\footnote{In our current implementation, we use CVX\cite{cvx} with package CVXQUAD\cite{FSP17} to solve the convex program stated in the first step of \Cref{alg:final_temp}.} of our \Cref{alg:final_temp}; we think building such an efficient practical solver is an important research direction.

In \Cref{app:experiments}, we provide experiments for other distributions, compare the performance of the PseudoPML approach implemented using our algorithm with a heuristic approximate PML algorithm~\cite{PJW17} and provide all the implementation details.

\section*{Acknowledgments}

We thank Alon Orlitsky and Yi Hao for helpful clarifications and discussions. 

\section*{Sources of Funding}

Researchers on this project were supported by a Microsoft Research Faculty Fellowship, NSF CAREER Award CCF-1844855, NSF Grant CCF-1955039, a Simons Investigator Award, a Google Faculty Research Award, an Amazon Research Award, a PayPal research gift, a Sloan Research Fellowship, a Stanford Data Science Scholarship and a Dantzig-Lieberman Operations Research Fellowship.
%\end{ack}

%\section*{Competing Interests}
%The authors declare no competing interests.

\newpage
\bibliographystyle{alpha}
\bibliography{PML}

\newcommand{\etalchar}[1]{$^{#1}$}
\begin{thebibliography}{ADM{\etalchar{+}}10}

\bibitem[ACSS20]{ACSS20}
Nima Anari, Moses Charikar, Kirankumar Shiragur, and Aaron Sidford.
\newblock The bethe and sinkhorn permanents of low rank matrices and
  implications for profile maximum likelihood, 2020.

\bibitem[ADM{\etalchar{+}}10]{ADMOP10}
J.~Acharya, H.~Das, H.~Mohimani, A.~Orlitsky, and S.~Pan.
\newblock Exact calculation of pattern probabilities.
\newblock In {\em 2010 IEEE International Symposium on Information Theory},
  pages 1498--1502, June 2010.

\bibitem[ADOS16]{ADOS16}
Jayadev Acharya, Hirakendu Das, Alon Orlitsky, and Ananda~Theertha Suresh.
\newblock A unified maximum likelihood approach for optimal distribution
  property estimation.
\newblock {\em CoRR}, abs/1611.02960, 2016.

\bibitem[AOST14]{AOST14}
Jayadev Acharya, Alon Orlitsky, Ananda~Theertha Suresh, and Himanshu Tyagi.
\newblock The complexity of estimating rényi entropy.
\newblock In {\em Proceedings of the Twenty-Sixth Annual ACM-SIAM Symposium on
  Discrete Algorithms}, pages 1855--1869, 2014.

\bibitem[AOST17]{AOST17}
Jayadev Acharya, Alon Orlitsky, Ananda~Theertha Suresh, and Himanshu Tyagi.
\newblock Estimating renyi entropy of discrete distributions.
\newblock {\em IEEE Trans. Inf. Theor.}, 63(1):38--56, January 2017.

\bibitem[AV20]{AW20}
Josh {Alman} and Virginia {Vassilevska Williams}.
\newblock {A Refined Laser Method and Faster Matrix Multiplication}.
\newblock {\em arXiv e-prints}, page arXiv:2010.05846, October 2020.

\bibitem[BF93]{BF93}
John Bunge and Michael Fitzpatrick.
\newblock Estimating the number of species: a review.
\newblock {\em Journal of the American Statistical Association},
  88(421):364--373, 1993.

\bibitem[BZLV16]{BZLV16}
Y.~{Bu}, S.~{Zou}, Y.~{Liang}, and V.~V. {Veeravalli}.
\newblock Estimation of kl divergence between large-alphabet distributions.
\newblock In {\em 2016 IEEE International Symposium on Information Theory
  (ISIT)}, pages 1118--1122, July 2016.

\bibitem[CCG{\etalchar{+}}12]{CCGLMCL12}
Robert~K Colwell, Anne Chao, Nicholas~J Gotelli, Shang-Yi Lin, Chang~Xuan Mao,
  Robin~L Chazdon, and John~T Longino.
\newblock Models and estimators linking individual-based and sample-based
  rarefaction, extrapolation and comparison of assemblages.
\newblock {\em Journal of plant ecology}, 5(1):3--21, 2012.

\bibitem[Cha84]{Chao84}
A~Chao.
\newblock Nonparametric estimation of the number of classes in a population.
  scandinavianjournal of statistics11, 265-270.
\newblock {\em Chao26511Scandinavian Journal of Statistics1984}, 1984.

\bibitem[CSS19a]{CSS19}
Moses Charikar, Kirankumar Shiragur, and Aaron Sidford.
\newblock Efficient profile maximum likelihood for universal symmetric property
  estimation.
\newblock In {\em Proceedings of the 51st Annual ACM SIGACT Symposium on Theory
  of Computing}, STOC 2019, pages 780--791, New York, NY, USA, 2019. ACM.

\bibitem[CSS19b]{CSS19pseudo}
Moses Charikar, Kirankumar Shiragur, and Aaron Sidford.
\newblock A general framework for symmetric property estimation.
\newblock In H.~Wallach, H.~Larochelle, A.~Beygelzimer, F.~d\textquotesingle
  Alch\'{e}-Buc, E.~Fox, and R.~Garnett, editors, {\em Advances in Neural
  Information Processing Systems 32}, pages 12447--12457. Curran Associates,
  Inc., 2019.

\bibitem[DS13]{DS13}
Timothy Daley and Andrew~D Smith.
\newblock Predicting the molecular complexity of sequencing libraries.
\newblock {\em Nature methods}, 10(4):325, 2013.

\bibitem[ET76]{ET76}
Bradley Efron and Ronald Thisted.
\newblock Estimating the number of unseen species: How many words did
  shakespeare know?
\newblock {\em Biometrika}, 63(3):435--447, 1976.

\bibitem[FSP17]{FSP17}
H.~{Fawzi}, J.~{Saunderson}, and P.~A. {Parrilo}.
\newblock {Semidefinite approximations of the matrix logarithm}.
\newblock {\em ArXiv e-prints}, May 2017.

\bibitem[F{\"u}r05]{Fur05}
Johannes F{\"u}rnkranz.
\newblock Web mining.
\newblock In {\em Data mining and knowledge discovery handbook}, pages
  899--920. Springer, 2005.

\bibitem[GB14]{cvx}
Michael Grant and Stephen Boyd.
\newblock {CVX}: Matlab software for disciplined convex programming, version
  2.1.
\newblock \url{http://cvxr.com/cvx}, March 2014.

\bibitem[GO13]{GO13}
Anupam Gupta and Ryan O’Donnell.
\newblock Lecture notes for cmu’s course on linear programming \&
  semidefinite programming.
\newblock
  \url{https://www.cs.cmu.edu/afs/cs.cmu.edu/academic/class/15859-f11/www/notes/lpsdp.pdf},
  November 2013.

\bibitem[GTPB07]{GTPB07}
Zhan Gao, Chi-hong Tseng, Zhiheng Pei, and Martin~J Blaser.
\newblock Molecular analysis of human forearm superficial skin bacterial biota.
\newblock {\em Proceedings of the National Academy of Sciences},
  104(8):2927--2932, 2007.

\bibitem[GW92]{GW92}
Harold~N Gabow and Herbert~H Westermann.
\newblock Forests, frames, and games: algorithms for matroid sums and
  applications.
\newblock {\em Algorithmica}, 7(1-6):465, 1992.

\bibitem[HHRB01]{HHRB01}
Jennifer~B Hughes, Jessica~J Hellmann, Taylor~H Ricketts, and Brendan~JM
  Bohannan.
\newblock Counting the uncountable: statistical approaches to estimating
  microbial diversity.
\newblock {\em Appl. Environ. Microbiol.}, 67(10):4399--4406, 2001.

\bibitem[HJW16]{HJW16}
Yanjun Han, Jiantao Jiao, and Tsachy Weissman.
\newblock Minimax estimation of {KL} divergence between discrete distributions.
\newblock {\em CoRR}, abs/1605.09124, 2016.

\bibitem[HJW18]{HJW18}
Yanjun Han, Jiantao Jiao, and Tsachy Weissman.
\newblock Local moment matching: A unified methodology for symmetric functional
  estimation and distribution estimation under wasserstein distance.
\newblock {\em arXiv preprint arXiv:1802.08405}, 2018.

\bibitem[HO19]{HO19}
Yi~{Hao} and Alon {Orlitsky}.
\newblock {The Broad Optimality of Profile Maximum Likelihood}.
\newblock {\em arXiv e-prints}, page arXiv:1906.03794, Jun 2019.

\bibitem[HO20]{HO20pentropy}
Yi~Hao and Alon Orlitsky.
\newblock Profile entropy: A fundamental measure for the learnability and
  compressibility of discrete distributions, 2020.

\bibitem[HS20]{HS20}
Yanjun Han and Kirankumar Shiragur.
\newblock The optimality of profile maximum likelihood in estimating sorted
  discrete distributions, 2020.

\bibitem[JHW16]{JHW16}
J.~Jiao, Y.~Han, and T.~Weissman.
\newblock Minimax estimation of the l1 distance.
\newblock In {\em 2016 IEEE International Symposium on Information Theory
  (ISIT)}, pages 750--754, July 2016.

\bibitem[JVHW15]{JVHW15}
J.~Jiao, K.~Venkat, Y.~Han, and T.~Weissman.
\newblock Minimax estimation of functionals of discrete distributions.
\newblock {\em IEEE Transactions on Information Theory}, 61(5):2835--2885, May
  2015.

\bibitem[Kar00]{Karger00}
David~R Karger.
\newblock Minimum cuts in near-linear time.
\newblock {\em Journal of the ACM (JACM)}, 47(1):46--76, 2000.

\bibitem[KLR99]{KLR99}
Ian Kroes, Paul~W Lepp, and David~A Relman.
\newblock Bacterial diversity within the human subgingival crevice.
\newblock {\em Proceedings of the National Academy of Sciences},
  96(25):14547--14552, 1999.

\bibitem[KS96]{Karger96}
David~R Karger and Clifford Stein.
\newblock A new approach to the minimum cut problem.
\newblock {\em Journal of the ACM (JACM)}, 43(4):601--640, 1996.

\bibitem[LG14]{Gall14}
Fran\c{c}ois Le~Gall.
\newblock Powers of tensors and fast matrix multiplication.
\newblock In {\em Proceedings of the 39th International Symposium on Symbolic
  and Algebraic Computation}, ISSAC '14, page 296–303, New York, NY, USA,
  2014. Association for Computing Machinery.

\bibitem[LTWZ13]{LMTWZ13}
L{\'a}szl{\'o}~Mikl{\'o}s Lov{\'a}sz, Carsten Thomassen, Yezhou Wu, and
  Cun-Quan Zhang.
\newblock Nowhere-zero 3-flows and modulo k-orientations.
\newblock {\em Journal of Combinatorial Theory, Series B}, 103(5):587--598,
  2013.

\bibitem[NW61]{nash1961edge}
C~St~JA Nash-Williams.
\newblock Edge-disjoint spanning trees of finite graphs.
\newblock {\em Journal of the London Mathematical Society}, 1(1):445--450,
  1961.

\bibitem[OSS{\etalchar{+}}04]{OSSVZ04}
A.~Orlitsky, S.~Sajama, N.~P. Santhanam, K.~Viswanathan, and Junan Zhang.
\newblock Algorithms for modeling distributions over large alphabets.
\newblock In {\em International Symposium on Information Theory, 2004. ISIT
  2004. Proceedings.}, pages 304--304, 2004.

\bibitem[OSW16]{OSW16}
Alon Orlitsky, Ananda~Theertha Suresh, and Yihong Wu.
\newblock Optimal prediction of the number of unseen species.
\newblock {\em Proceedings of the National Academy of Sciences},
  113(47):13283--13288, 2016.

\bibitem[PBG{\etalchar{+}}01]{PBGELLSD01}
Bruce~J Paster, Susan~K Boches, Jamie~L Galvin, Rebecca~E Ericson, Carol~N Lau,
  Valerie~A Levanos, Ashish Sahasrabudhe, and Floyd~E Dewhirst.
\newblock Bacterial diversity in human subgingival plaque.
\newblock {\em Journal of bacteriology}, 183(12):3770--3783, 2001.

\bibitem[PJW17]{PJW17}
D.~S. {Pavlichin}, J.~{Jiao}, and T.~{Weissman}.
\newblock {Approximate Profile Maximum Likelihood}.
\newblock {\em ArXiv e-prints}, December 2017.

\bibitem[RCS{\etalchar{+}}09]{RCSWTKRWC09}
Harlan~S Robins, Paulo~V Campregher, Santosh~K Srivastava, Abigail Wacher,
  Cameron~J Turtle, Orsalem Kahsai, Stanley~R Riddell, Edus~H Warren, and
  Christopher~S Carlson.
\newblock Comprehensive assessment of t-cell receptor $\beta$-chain diversity
  in $\alpha$$\beta$ t cells.
\newblock {\em Blood}, 114(19):4099--4107, 2009.

\bibitem[TE87]{TE87}
Ronald Thisted and Bradley Efron.
\newblock Did shakespeare write a newly-discovered poem?
\newblock {\em Biometrika}, 74(3):445--455, 1987.

\bibitem[Tho14]{thomassen2014graph}
Carsten Thomassen.
\newblock Graph factors modulo k.
\newblock {\em Journal of Combinatorial Theory, Series B}, 106:174--177, 2014.

\bibitem[Von12]{Von12}
Pascal~O. Vontobel.
\newblock The bethe approximation of the pattern maximum likelihood
  distribution.
\newblock pages 2012--2016, 07 2012.

\bibitem[Von14]{Von14}
P.~O. Vontobel.
\newblock The bethe and sinkhorn approximations of the pattern maximum
  likelihood estimate and their connections to the valiant-valiant estimate.
\newblock In {\em 2014 Information Theory and Applications Workshop (ITA)},
  pages 1--10, Feb 2014.

\bibitem[VV11a]{VV11b}
G.~Valiant and P.~Valiant.
\newblock The power of linear estimators.
\newblock In {\em 2011 IEEE 52nd Annual Symposium on Foundations of Computer
  Science}, pages 403--412, Oct 2011.

\bibitem[VV11b]{VV11a}
Gregory Valiant and Paul Valiant.
\newblock Estimating the unseen: An n/log(n)-sample estimator for entropy and
  support size, shown optimal via new clts.
\newblock In {\em Proceedings of the Forty-third Annual ACM Symposium on Theory
  of Computing}, STOC '11, pages 685--694, New York, NY, USA, 2011. ACM.

\bibitem[WD14]{WX14}
David~P. Williamson and Xiaobo Ding.
\newblock Orie 6300 mathematical programming i: Lecture 12.
\newblock
  \url{https://people.orie.cornell.edu/dpw/orie6300/Lectures/lec12.pdf},
  October 2014.

\bibitem[Wil12]{Williams12}
Virginia~Vassilevska Williams.
\newblock Multiplying matrices faster than coppersmith-winograd.
\newblock In {\em Proceedings of the Forty-Fourth Annual ACM Symposium on
  Theory of Computing}, STOC '12, page 887–898, New York, NY, USA, 2012.
  Association for Computing Machinery.

\bibitem[WY15]{WY15}
Y.~{Wu} and P.~{Yang}.
\newblock {Chebyshev polynomials, moment matching, and optimal estimation of
  the unseen}.
\newblock {\em ArXiv e-prints}, April 2015.

\bibitem[WY16]{WY16}
Y.~Wu and P.~Yang.
\newblock Minimax rates of entropy estimation on large alphabets via best
  polynomial approximation.
\newblock {\em IEEE Transactions on Information Theory}, 62(6):3702--3720, June
  2016.

\bibitem[ZVV{\etalchar{+}}16]{ZVVKCSLSDM16}
James Zou, Gregory Valiant, Paul Valiant, Konrad Karczewski, Siu~On Chan,
  Kaitlin Samocha, Monkol Lek, Shamil Sunyaev, Mark Daly, and Daniel~G.
  MacArthur.
\newblock Quantifying unobserved protein-coding variants in human populations
  provides a roadmap for large-scale sequencing projects.
\newblock {\em Nature Communications}, 7:13293 EP--, Oct 2016.

\end{thebibliography}
%\printbibliography
\appendix
\newpage
%!TEX root = neurips_2020.tex
\newpage

\section{Remaining Proofs from \Cref{sec:mainfirst}}\label{app:approxpml}

Here we provide proofs for all the results in \Cref{sec:mainfirst} that were excluded in the main paper. For each of these results we dedicate a subsection that provides further details. Combining all these results from different subsections, in \Cref{app:thmgeneralone} we provide the proof for our main result (\Cref{thm:resultmainone}).

\subsection{Properties of Convex Program and Proof of \Cref{lem:sparse}}\label{app:sparse}
%!TEX root = neurips_2020.tex
%\section{Properties of the Convex Program and proof of \Cref{lem:sparse}}\label{app:properties}
Here we prove important properties of our convex program. For convenience, we define the negative log of function $\bg(\mx)$,
\begin{equation}\label{eq:loggx}
\bff(\mx)\defeq \sum_{i\in[1,\ell],j\in[0,k]}\left[-\mc_{ij}\mx_{ij}+\mx_{ij}\log\mx_{ij}\right]-\sum_{i\in[1,\ell]}[\mx\vones]_{i}\log[\mx\vones]_{i}= -\log \bg(\mx)~.
\end{equation}
In the remainder we prove and state interesting properties of this function that helps us construct sparse approximate solutions. We start by recalling properties showed in \cite{CSS19}.
\begin{lemma}[Lemma 4.16 in \cite{CSS19}]
	Function $\bff(\mx)$ is convex in $\mx$.
\end{lemma}
\begin{theorem}[Theorem 4.17 in \cite{CSS19}]\label{thm:convrt}
	Given a profile $\phi \in \Phi^{n}$ with $k$ distinct frequencies, the optimization problem $\min_{\mx \in \bZfrac}\bff(\mx)$
	%$\max_{\bS \in \bZfrac}\log \bg(\bS)$ 
	can be solved in time $\otilde(k^2|\bR|)$.
\end{theorem}
The function $\bff(\mx)$ is separable in each row and we define following notation to capture it.
$$\bff_i(\mx_i)\defeq \sum_{j\in[0,k]}\left[-\mc_{ij}\mx_{ij}+\mx_{ij}\log\mx_{ij}\right]-[\mx\vones]_{i}\log\left([\mx\vones]_{i}\right) \quad \text{ and  }\quad \bff(\mx)=\sum_{i \in [1,\ell]}\bff_i(\mx_i)~.$$
The function $\bff_i(\mx_i)$ defined above is $1$-homogeneous and is formally shown next.
\begin{lemma}\label{lem:hom}
	For any fixed vector $c\in \R^{[0,k]}$, the function $\fnh(v) = \sum_{j\in[0,k]}\left[\vecc_{j} v_j+v_j\log v_j\right] -v^{\top}\1 \log v^{\top}\1$ is $1$-homogeneous, that is, $\fnh(\alpha\cdot v) = \alpha\cdot\fnh(v)$  for all $v \in \R_{\geq 0}^{[0,k]}$ and $\alpha \in \R_{\geq 0}$.
	\end{lemma}
\begin{proof}
	Consider any vector $v \in \R_{\geq 0}^{k+1}$ and scalar $\alpha \in \R_{\geq 0}$ we have,
	\begin{align*}
	\fnh(\alpha\cdot v) & =\sum_{j\in[0,k]}\left[\vecc_{j} (\alpha v_j)+(\alpha v_j)\log (\alpha v_j)\right]-(\alpha v)^{\top}\1 \log (\alpha v)^{\top}\1,\\
	& =\sum_{j\in[0,k]}\left[\vecc_{j} (\alpha v_j)+\alpha v_j\log v_j + \alpha v_j\log \alpha \right]-(\alpha v)^{\top}\1 \log v^{\top}\1 - (\alpha v)^{\top}\1 \log \alpha,\\
	& = \sum_{j\in[0,k]}\left[\vecc_{j} (\alpha v_j)+\alpha v_j\log v_j \right]-\alpha   v^{\top}\1 \log v^{\top}\1 = \alpha\cdot\fnh(v)~.\\
	\end{align*}
	The above derivation satisfies the conditions of the lemma and we conclude the proof.
\end{proof}
In the remainder of this section, we provide the proof of \Cref{lem:sparse} and the description of the algorithm $\sparse$ is included inside the proof. The \Cref{lem:sparse} in the notation of $\bff(\cdot)$ can be equivalently written as follows.
\begin{lemma}[\Cref{lem:sparse}]\label{lem:sparseeq}
	For any $\mx \in \bZfrac$, the algorithm $\sparse(\mx)$ runs in $\otilde(|\bR|~k^{\omega})$ time and returns a solution $\mx' \in \bZfrac$ such that $\bff(\mx') \leq \bff(\mx)$ and $\big|\{i \in [1,\ell]~|~[\mx'\1]_{i}  > 0\}\big| \leq k+1$.
\end{lemma}
\begin{proof}
	Let $\ell\defeq |\bR|$ and fix $\mx \in \bZfrac$, consider the following solution $\mx'_{i}=\alpha_{i}\mx_{i}$ for all $i\in [1,\ell]$, where $\alpha \in \R_{\geq 0}^{[1,\ell]}$ and $\mx_{i},\mx'_{i}$ denote the vectors corresponding to the $i$'th row of matrices $\mx,\mx'$ respectively. By \Cref{lem:hom}, each function $\bff_{i}(\mx_{i})$ is 1-homogeneous and we get,
	\begin{align*}
	\bff(\mx')&=\sum_{i\in [1,\ell]} \bff_{i}(\mx'_{i})=\sum_{i\in [1,\ell]}  \bff_{i}(\alpha_{i}\mx_{i})=\sum_{i\in [1,\ell]} \alpha_{i} \bff_{i}(\mx_{i})~.
	\end{align*}
	Let $\alpha \in \R_{\geq 0}^{[1,\ell]}$ be such that the following conditions hold,
	\begin{equation}\label{eq:falpha}
	\sum_{i \in [1,\ell]} \alpha_{i} \mx_{i,j} = \phi_{j} \text{ for all } j\in [1,k] \text{ and }\sum_{i \in [1,\ell]}\alpha_{i} \ri [\mx \vones]_{i}\leq 1~.
	\end{equation}
	For the above set of equations, the solution $\alpha=\vones$ is feasible as $\mx \in \bZfrac$. Further for any $\alpha$ satisfying the above inequalities, the corresponding matrix $\mx'$ satisfies,
	$$\sum_{i \in [1,\ell]} \mx'_{i,j}=\sum_{i \in [1,\ell]} \alpha_{i} \mx_{i,j} = \phi_{j} \text{ for all } j\in [1,k] \text{ and }\sum_{i \in [1,\ell]}\ri [\mx' \vones]_{i}=\sum_{i \in [1,\ell]}\alpha_{i} \ri [\mx \vones]_{i}\leq 1~.$$
	Therefore $\mx' \in \bZfrac$ for all $\alpha \in \R_{\geq 0}^{[1,\ell]}$ that satisfy \Cref{eq:falpha}. In the remainder of the proof we find a sparse $\alpha$ that satisfies the conditions of the lemma. 
	
	Consider the following linear program.
	\begin{align*}
	\min{\alpha \in \R_{\geq 0}^{[1,\ell]}} & \sum_{i \in [1,\ell]} \alpha_{i} \bff_{i}(\mx_{i})~.\\
	\text{ such that, }& \sum_{i \in [1,\ell]} \alpha_{i} \mx_{i,j} = \phi_{j} \text{ for all } j\in [1,k] \text{ and }\sum_{i \in [1,\ell]}\alpha_{i} \ri [\mx \vones]_{i}\leq 1~.
	\end{align*}
	Note in the above optimization problem we fix $\mx \in \bZfrac$ and optimize over $\alpha$. Any basic feasible solution (BFS) $\alpha^{*}$ to the above LP, satisfies $|\{i\in [1,\ell]~|~ \alpha^{*}_{i}>0\}| \leq k+1$ as there are at most $k+1$ non-trivial constraints. Suppose we find a basic feasible solution $\alpha^*$ such that the corresponding matrix $\mx'_{i} = \alpha^{*}_{i}\mx_{i}$ for all $i \in [1,\ell]$ satisfies $\bff(\mx') \leq \bff(\mx)$, then such a matrix $\mx'$ is the desired solution that satisfies the conditions of the lemma. Therefore in the remainder of the proof, we discuss the running time to find such a BFS given a feasible solution to the LP. Finding a BFS to a linear program is quite standard; please refer to lecture notes~\cite{WX14,GO13} for further details. For completeness, in the following we provide an algorithm to find a desired BFS and analyze its running time.
	
	Leveraging these insights, we design the following iterative algorithm. In each iteration $i$ we maintain a set  $S_i \subseteq \R^{k + 1}$ of $1\leq k_i \leq k+1$ linearly independent rows of matrix $\mx$. We update the solution $\alpha$ and try to set a non-zero coordinate of it to value zero while not increasing the objective.  
	Our algorithm starts with $k_i = 1$ and $S_i$ to be the set containing an arbitrary row of $\mx$ in iteration $i=1$. The next iteration is computed by considering an arbitrary row $r$ of matrix $\mx$ that corresponds to a non-zero coordinate in $\alpha$. Letting $\ma_i \in \R^{(k+1) \times k_i}$ be the matrix where the columns are the vectors in $S_i$ we then consider the linear system $\ma_i^\top \ma_i x = r$. Whether or not there is such a solution can be computed in $O(k^{\omega})$, where $\omega < 2.373$ is the matrix multiplication constant \cite{Williams12,Gall14,AW20} using fast matrix multiplication as in this time we can form the $(k + 1) \times (k + 1)$ matrix $\ma_i^\top \ma_i$ directly and then invert it. If this system has no solution we let $S_{i + 1} = S_i \cup r$ and proceed to the next iteration as the lack of a solution proves that $S_i \cup r$ are linearly independent (as $S_i$ is linearly independent). Otherwise, we consider the vector $\alpha'$ in the null space of the transpose of $\mx$ formed by setting $\alpha'_i$ to the value of $x_j$ for the associated rows and setting $\alpha'_i$ for the row corresponding to row $r$ to be $-1$. As $x$ is a solution to  $\ma_i^\top \ma_i x = r$, clearly $\mx^\top \alpha' = 0$. Now consider the solution $\alpha+c\alpha'$ for some scaling $c$. Since the objective and constraints are linear, there exists a direction, that is, sign of $c$ such that the objective is non-increasing and the solution $\alpha+c\alpha'$ satisfies all the constraints (\Cref{eq:falpha}). We start with $c=0$ and keep  increasing it in the direction where the objective in non-increasing till one of the following two conditions hold: either a new coordinate in the solution $\alpha+c\alpha'$ becomes zero or the objective value of the LP is infinity. In the first case, we update our current solution $\alpha$ to $\alpha+c\alpha'$ and repeat the procedure. As the goal our algorithm is to find a sparse solution, we fix the co-ordinates in $\alpha$ that have value zero and never change (or consider) them in the later iterations of our algorithm. We repeat this procedure till all the non-zero co-ordinates in $\alpha$ are considered at least once and the solution $\alpha$ returned at the end corresponds to a BFS that satisfies the desired conditions. As the total number of rows is at most $\ell$, our algorithm has at most $\ell$ iterations and each iteration takes only $O(k^{\omega})$ time (note that we only update $O(k)$ coordinates in each iteration). Therefore the final running time of the algorithm $\sparse$ is $\otilde(\ell k^{\omega})$ time and we conclude the proof. 

\end{proof}

\subsection{Remaining Parts of the Proof for \Cref{thm:matrixround}}\label{app:matrixround}

We first finish the proof of \cref{lem:mod}. That only leaves us with proving \cref{lem:connected-alg}.
\begin{proof}[Proof of \cref{lem:mod} in the general case]
	Since the input graph is arbitrary, we have no guarantee about edge-connectivity. We will show that we can remove $O(k\card{V})$ edges from $G$ so that the remaining subgraph is a vertex-disjoint union of $6k$-edge-connected induced subgraphs. To do this, look at the connected components of $G$. Either they are all $6k$-edge-connected or at least one of them has a cut with $<6k$ edges. Moreover we can check this in polynomial time (and find violating cuts if there are any) by a global minimum cut algorithm \cite{Karger00}. If a component is not $6k$-edge-connected, remove all edges of the small cut, and repeat. Every time we remove the edges of a cut, the number of connected components increases by $1$, so this can go on for at most $O(\card{V})$ iterations. In each iteration, at most $6k$ edges are removed, so the total number of removed edges is $O(k\card{V})$.
	
	So by removing $O(k\card{V})$ edges, we have transformed $G$ into a vertex-disjoint union of $6k$-edge-connected graphs. We simply apply the already-proved case of \cref{lem:mod} to each of these components to get our desired result for the original graph $G$.
\end{proof}
In the remainder of this section we prove \cref{lem:connected-alg}. We do this by showing how to make the proof of \cref{lem:connected} due to \cite{thomassen2014graph} algorithmic. \cite{thomassen2014graph} reduced \cref{lem:connected} to an earlier result by \cite{LMTWZ13} which we state below.
\begin{lemma}[{\cite[Theorem 1.12]{LMTWZ13}}]\label{lem:orientation}
	Let $k\geq 3$ be an odd integer and $G=(V,E)$ a $(3k-3)$-edge connected undirected graph. For any given $\beta:V\to \set{0,\dots,k-1}$ where $\sum_{v}\beta(v)\equiv 0\pmod{k}$, there is an orientation of $G$ which makes $\deg_{\mathrm{out}}(v)-\deg_{\mathrm{in}}(v)$ equal to $\beta(v)$ modulo $k$ for every vertex $v$.
\end{lemma}
Here an orientation is an assignment of one of the two possible directions to each edge, and $\deg_{\mathrm{out}}$ and $\deg_{\mathrm{in}}$ count outgoing and incoming edges of a vertex in such an orientation. We simply note that the reduction of \cref{lem:connected} to \cref{lem:orientation}, as stated in \cite{thomassen2014graph}, is already efficient. This is done by a simple transformation on $f$ from \cref{lem:connected} to get $\beta$, and at the end a subgraph is extracted from an orientation by considering edges oriented from one side to the other. Since the reduction is efficient, we simply need to prove \cref{lem:orientation} can be made efficient.
\begin{lemma}\label{lem:orientation-alg}
	There is a polynomial time algorithm that outputs the orientation of \cref{lem:orientation}.
\end{lemma}

To obtain this algorithm, our strategy is to make the steps of the proof presented in \cite{LMTWZ13} (efficiently) constructive. \cite{LMTWZ13} prove \cref{lem:orientation} by generalizing the statement and using a clever induction. To state this generalization, we need a definition from \cite{LMTWZ13}.
\begin{defn}[\cite{LMTWZ13}]\label{def:tau}
	Suppose that $k$ is an odd integer, and $G=(V, E)$ is an undirected graph. For a given function $\beta:V\to \set{0,\dots,k-1}$, we define a set function $\tau:2^{V}\to \set{0,\pm 1,\dots,\pm k}$ by the following congruences
	\begin{align*}
		\tau(S) & \equiv\sum_{v\in S} \beta(S) \pmod{k}\\
		\tau(S) & \equiv\sum_{v\in S} \deg(S) \pmod{2}\\
	\end{align*}
\end{defn}
The two given congruences uniquely determine $\tau(S)$ modulo $2k$; this in turn is a unique element of $\set{0,\pm 1,\dots,\pm k}$, except for $k$ and $-k$ which are the same value modulo $2k$. The choice of which value to take in this case is largely irrelevant, as we will mostly be dealing with $\abs{\tau(\cdot)}$. Note that $\tau(S)$ is the same, modulo $2k$, as the number of edges going from $S$ to $S^c$ minus the number of edges going from $S^c$ to $S$ in any valid orientation as promised by \cref{lem:orientation}.

The definition of $\tau$ is used to give a generalization of \cref{lem:orientation} that is proved by induction.
\begin{lemma}[{\cite[Theorem 3.1]{LMTWZ13}}]\label{lem:induct}
	Let $k$ be an odd integer, $G=(V, E)$ an undirected graph on at least $3$ vertices, and $\beta:V\to\set{0,\dots,k-1}$ be such that $\sum_v \beta(v)\equiv 0\pmod{k}$. Let $z_0$ be a ``special'' vertex of $G$ whose adjacent edges are already pre-oriented in a specified way. Assume that $\tau$ is defined as in \cref{def:tau} and $V_0=\set{v\in V-\set{z_0}\mid \tau(\set{v})=0}$; let $v_0$ be a vertex of minimum degree in $V_0$. If the following conditions are satisfied, then there is an orientation of edges, matching the pre-orientation of $z_0$, for which $\deg_{\mathrm{out}}(v)-\deg_{\mathrm{in}}(v)\equiv \beta(v) \pmod{k}$ for every $v$.
	\begin{enumerate}
		\item $\deg(z_0)\leq (2k-2)+\abs{\tau(\set{z_0})}$,
		\item $\card{E(S, S^c)}\geq (2k-2)+\card{\tau(S)}$ for every set $S$ where $z_0\notin S$, and $S\neq \emptyset, \set{v_0}, V-\set{z_0}$.
	\end{enumerate}
\end{lemma}
Here $E(S,S^c)$ is the set of edges between $S$ and $S^c$. Note that we always have $\abs{\tau(\cdot)}\leq k$. So a $(3k-3)$-edge-connected graph automatically satisfies condition 2 in \cref{lem:induct}. \Cref{lem:connected} is proved by adding an isolated vertex $z_0$ and setting $\beta(z_0)=0$, for which condition 1 is automatically satisfied.

The reason behind this generalization is the ability to prove it by induction. The authors of \cite{LMTWZ13} state this induction in the form of proof by contradiction. They consider a minimal counterexample, and argue the existence of a smaller counterexample. We do not state all of their proof again here, but note that all processes used to produce smaller counterexamples are readily efficiently implementable, except for one. In the proof of Theorem 3.1 in \cite{LMTWZ13}, in Claim 1, the authors argue that for non-singleton $S$ the inequality in condition 2 of \cref{lem:induct} cannot be strict, or else the size of the problem can be reduced. They formally prove that a smallest counterexample must satisfy for $\card{S}\geq 2$,
\begin{equation}\label{eq:falsify}  \card{E(S, S^c)}\geq 2k+\abs{\tau(S)}>(2k-2)+\abs{\tau(S)}. \end{equation}
In case a non-singleton does not satisfy the above inequality, the authors produce two smaller instances, once by contracting $S$ into a single vertex, and once by contracting $S^c$, and combining the resulting orientations together for all of $G$. The main barrier in making this into an efficient algorithm is \emph{finding} the set $S$ that violates the inequality. A priori, it might seem like an exhaustive search over all subsets $S$ is needed, but we show that this is not the case.

We now show how to make this part algorithmic.
\begin{lemma}
	Suppose that the graph $G$ satisfies the conditions of \cref{lem:induct}. Then there is a polynomial time algorithm which produces a list of sets $S_1,\dots,S_m$ for a polynomially bounded $m$, such that any violation of \cref{eq:falsify} must happen for some $S_i$.
\end{lemma}
\begin{proof}
	Our high-level strategy is to use the fact that condition 2 of \cref{lem:induct} implies $G$ is already sufficiently edge-connected. If $z_0, v_0$ did not exist, condition 2 would imply that $G$ is $(2k-2)$-edge-connected. On the other hand any violation of \cref{eq:falsify} can only happen when $\card{E(S, S^c)}<2k+k=3k$. So it would be enough to simply produce a list of all near-minimum-cuts $S$ with $\card{E(S, S^c)}< 3k$. If $G$ was $(2k-2)$-edge-connected, we could appeal to results of \cite{Karger96}, who proved that for any constant $\alpha$, the number of cuts of size at most $\alpha$ times the minimum cut is polynomially bounded and all of them can be efficiently enumerated.
	
	The one caveat is the existence of $v_0, z_0$, which might make $G$ not $(2k-2)$-edge-connected. Note that the only cuts that can potentially be ``small'' are the singletons $\set{v_0}, \set{z_0}$. We can solve this problem by contracting the graph. We enumerate over the edges $e_1, e_2$ that are adjacent to $v_0, z_0$, and for every choice of $e_1, e_2$, we produce a new graph by contracting the endpoints of $e_1$ followed by contracting the endpoints of $e_2$. If a cut $(S, S^c)$ does not have $v_0, z_0$ as a singleton on either side, there must be a choice of $e_1, e_2$ that do not cross the cut, which means that the cut ``survives'' the contraction. Note that the contracted graph is always $(2k-2)$-edge-connected, so we can proceed as before and produce a list of all of its cuts of size $<3k$. Taking the union of the list of all such cuts for all choices of $e_1, e_2$ produces the desired list we are seeking.
\end{proof}
We remark that a simple modification of our proof also shows that checking conditions 1 and 2 of \cref{lem:induct} can be done in polynomial time.

\subsection{Simplification and Details on \Cref{lem:create}}\label{app:create}
%!TEX root = neurips_2020.tex
%\section{Proof of \Cref{lem:create}}\label{app:createnewrows}
Here we state the lemma that captures the guarantees of the algorithm $\create$ from \cite{ACSS20}. We later apply this lemma in a specific setting where the conditions of \Cref{lem:create} are met and provide its proof.
\renewcommand{\bttpo}{(\btt+1)}
\renewcommand{\bSp}{\textbf{\ma}}
\renewcommand{\bS}{\textbf{\mb}}
\renewcommand{\boo}{\ell}
\renewcommand{\ztbtt}{[0,\btt]}

For a given profile $\phi$, the algorithm $\create$ takes input $(\bSp,\bS,\bR)$ and creates a solution pair $(\bS',\bR')$ that satisfy the following lemma.

\begin{lemma}\label{lem:createoriginal}
	Given a profile $\phi \in \Phi^n$ with $k$ distinct frequencies, a probability discretization set $\bR$ and matrices $\ma,\mb \in \R^{[\boo]\times \ztbtt}$ that satisfy: $\bSp \in \bZfrac$ and $\bS_{i,j} \leq \bSp_{i,j}$ for all $i\in [\boo]$ and $j\in \ztbtt$. There exists an algorithm that outputs a probability discretization set $\bR'$ and $\bSp'\in \R^{[\boo+\newk]\times \ztbtt}$ that satisfy the following guarantees,
	\begin{enumerate}[noitemsep,nosep,leftmargin=*]
		\item $\sum_{j\in \ztbtt}\bSp'_{i,j}=\sum_{j\in \ztbtt}\bS_{i,j}$ for all $i\in [\boo]$.
		\item For any $i \in [\boo+1,\boo+\newk]$, let $j \in \ztbtt$ be such that $i=\boo+1+j$ then $\bSp'_{\boo+1+j,j'}=0$ for all $j' \in \ztbtt$ and $j' \neq j$. (Diagonal Structure)
		\item For any $i\in [\boo+1,\boo+\newk]$, let $j \in \ztbtt$ be such that $i=\boo+1+j$, then $\sum_{j'\in \ztbtt}\bSp'_{i,j'}=\bSp'_{\boo+1+j,j}=\phi_{j}-\sum_{i' \in [\boo]}\bS_{i',j}$.
		\item $\bSp' \in \textbf{Z}^{\phi,frac}_{\bR'}$ and $\sum_{i \in [\boo+\newk]} \sum_{j \in \ztbtt} \bSp'_{i,j}=\sum_{i \in [\boo]} \sum_{j \in \ztbtt} \bSp_{i,j}$.
		%$\sum_{i\in [1,\boo+\btt]}\pvec_{i} \left( \sum_{j\in \ztbtt}\bSp'_{i,j} \right) \leq 1$.
		\item Let $\alpha_{i}\defeq \sum_{j\in \ztbtt}\bSp_{i,j} - \sum_{j\in \ztbtt}\bS_{i,j}$ for all $i\in [\boo]$ and $\logparam \defeq \max(\sum_{i \in [\boo]}(\bSp \onevec)_{i} , \boo \times \btt)$, then 
		$\bg(\bSp') \geq \expo{-O\left(  \sum_{i\in [\boo]} \alpha_{i}\log \logparam \right)} \bg(\bSp)~.$
		\item For any $j \in \ztbtt$, the new level sets have probability value equal to,
		$\pvec_{\boo+1+j}=\frac{\sum_{i \in \otboo}(\bSp_{ij}-\bS_{ij})\pvec_{i}}{\sum_{i \in \otboo}(\bSp_{ij}-\bS_{ij})}$.
	\end{enumerate}
\end{lemma}
W are now ready to provide the proof of \Cref{lem:create}
\begin{proof}[Proof of \Cref{lem:create}]
	By \Cref{lem:createoriginal}, we get a matrix $\ma' \in \R^{[\boo+\newk]\times \ztbtt}$ that satisfies $\ma' \in \textbf{Z}^{\phi,frac}_{\bR'}$ (guarantee 4 in \Cref{lem:createoriginal}) and $\bg(\bSp') \geq \expo{-O\left(  \sum_{i\in [\boo]} \alpha_{i}\log \logparam \right)} \bg(\bSp)$, where $\alpha_{i}\defeq \sum_{j\in \ztbtt}\bSp_{i,j} - \sum_{j\in \ztbtt}\bS_{i,j}$ for all $i\in [\boo]$ and $\logparam \defeq \max(\sum_{i \in [\boo]}(\bSp \onevec)_{i} , \boo \times \btt)$. 
	
	To prove the lemma we need to show two things: $\ma' \in \textbf{Z}^{\phi}_{\bR'}$ and $\bg(\ma') \geq \expo{-O\left(t\log n\right)}\bg(\ma)$. We start with the proof of the first expression. Note that $\ma' \in \textbf{Z}^{\phi,frac}_{\bR'}$ and we need to show that $\ma'$ has all integral row sums. For $i \in [\ell]$, the $i$'th row sum, that is $[\ma'\vones]_{i}$ is integral by combining guarantee 1 of \Cref{lem:createoriginal} and $[\mb\vones]_{i} \in \Z$ (condition of our current lemma). For $i\in [\boo+1,\boo+\newk]$, $[\ma'\vones]_{i}=\phi_{j}-[\mb^\top \vones]_{j}$ (guarantee 3 of \Cref{lem:createoriginal}) and the $i$'th row sum is integral because $[\mb^\top \vones]_{j} \in \Z$ (condition of our current lemma) and $[\mb^\top \vones]_{j} \leq [\ma^\top \vones]_{j} \leq \phi_j$.
	
	We now shift our attention to the second expression, that is $\bg(\ma') \geq \expo{-O\left(t\log n\right)}\bg(\ma)$. We prove this inequality by providing bounds on the parameters $\Delta$, $\alpha_{i}$. Observe that $\Delta \leq 1/\rmin + \ell k \leq 1/\rmin+k(k+1) \leq O(n^2)$ because $\ma \in \bZfrac$ and therefore satisfies $\sum_{i \in [1,k+1]} \ri[\ma\vones]_{i} \leq 1$ that further implies $\sum_{i \in [1,k+1]} [\ma'\vones]_{i} \leq 1/\rmin \leq 2n^2$ (see the definition of probability discretization). In the second inequality for the bound on $\Delta$ we used $\ell \leq k+1$, as without loss of generality the number of probability values in $|\bR|$ can be assumed to be at most $k+1$ (because of the sparsity lemma \Cref{lem:sparse}) and the actual size of $|\bR|$ only reflects in the running time. Now note that $\sum_{i\in [k+1]} \alpha_{i}=\sum_{i\in [\ell], j \in [0,k]} (\ma_{ij}-\mb_{ij})\leq t$ because of the condition of the lemma. Combining the analysis for $\Delta$ and $\alpha_{i}$, we get $\bg(\ma') \geq \expo{-O\left(t\log n\right)}\bg(\ma)$ and we conclude the proof. 
\end{proof}

\subsection{Proof of \Cref{thm:generalone} and \Cref{thm:resultmainone}}\label{app:thmgeneralone}
\renewcommand{\bS}{\textbf{S}}
\renewcommand{\bSp}{\textbf{S}'}
Here we provide the proof of \Cref{thm:generalone}, that provides the guarantees of our first rounding algorithm (Algorithm 1) for any probability descritization set $\bR$. Later we choose this discretization set carefully to prove our main theorem (\Cref{thm:resultmainone}).
\begin{proof}[Proof of \Cref{thm:generalone}]	
	By \Cref{lem:maximizer}, the Step 1 returns a solution $\bSp \in \bZfrac$ that satisfies,
	$\cphi \cdot \bg(\bSp)\geq \expo{\bigO{-k \log n}} \max_{\bq \in \dsimplex}\probpml(\bq,\phi)$.
	By \Cref{lem:sparse}, the Step 2 takes input $\bS'$ and outputs $\bS'' \in \bZfrac$ such that $\bg(\bS'') \geq \bg(\bS')$ and $\big|\{i \in [\ell]~|~[\bS''\1]_{i}  > 0\}\big| \leq k+1$. As the matrix $\bS''$ has at most $k+1$ non-zero rows and columns, by \Cref{thm:matrixround} the Step 3 returns a matrix $\mb''$ that satisfies: $\mb''_{ij} \leq \bS''_{ij} ~\forall ~i \in [\ell],j \in [0,k]$, $\mb''\1\in \Z^\ell$, $\mb''^\top \1\in \Z^{[0,k]}$ and $\sum_{i\in [\ell], j \in [0,k]} (\bS''_{ij}-\mb''_{ij})\leq O(k)$. The matrices $\bS''$ and $\mb''$ satisfy the conditions of \Cref{lem:create} with parameter $t=O(k)$ and the algorithm $\create$ returns a solution $(\bSext,\bRext)$ such that $\bSext \in \bZext$ and $\bg(\bSext) \geq \exp(-O(k\log n))\bg(\bS'')$. Further substituting $\bg(\bS'') \geq \bg(\bS')$ from earlier (Step 2) we get, $\bg(\bSext) \geq \exp(-O(k\log n))\bg(\bS')$. As $\bSext \in \bZext$, by \Cref{lem:associateddist} the associated distribution $\bp'$ satisfies $\probpml(\bp',\phi) \geq \exp(-O(k\log n)) \cphi \cdot \bg(\bSext)\geq \exp(-O(k\log n))\cphi \cdot \bg(\bS')$. Further combined with inequality $\cphi \cdot \bg(\bSp)\geq \expo{\bigO{-k \log n}} \max_{\bq \in \dsimplex}\probpml(\bq,\phi)$ (Step 1) we get,
	$$\probpml(\bp',\phi)\geq \expo{\bigO{-k \log n}} \max_{\bq \in \dsimplex}\probpml(\bq,\phi)~.$$
	All the steps in our algorithm run in polynomial time and we conclude the proof.
\end{proof}

\begin{proof}[Proof of \Cref{thm:resultmainone}]
Choose $\bR$ with parameters $\alpha=k\log n/n$ and $|\bR|=\ell=O(n/k)$ in \Cref{lem:probdisc} and we get that $\max_{\bq\in \dsimplex} \probpml(\bq,\phi) \geq \expo{-k\log n}\max_{\bp\in \simplex}\probpml(\bp,\phi)$.
As the $|\bR|$ is polynomial in $n$, the previous inequality combined with \Cref{thm:generalone} proves our theorem.
\end{proof}	
\section{PseudoPML Approach, Remaining Proofs from \Cref{sec:practical} and Experiments}\label{app:pseudoall}
Here we provide all the details regarding the PseudoPML approach. PseudoPML also known as TrucatedPML was introduced independently in \cite{CSS19pseudo} and \cite{HO19}. In \Cref{app:pseudoapproxpml}, we provide the proof for the guarantees achieved by our second rounding algorithm (\Cref{thm:secmain}) that in turn helps us prove \Cref{thm:resultmaintwo}. In \Cref{app:notation}, we provide notations and definitions related to the PseudoPML approach. In \Cref{app:pseduo}, we provide the proof of \Cref{lem:pseudomain}. Finally in \Cref{app:experiments}, we provide the remaining experimental results and the details of our implementation.

\subsection{Proof of \Cref{thm:secmain} and \Cref{thm:resultmaintwo}}\label{app:pseudoapproxpml}
Here we provide the proof of \Cref{thm:secmain} that provides the guarantees satisfied by our second approximate PML algorithm. Further using this theorem , we provide the proof for \Cref{thm:resultmaintwo}.
\renewcommand{\bS}{\textbf{S}}

\begin{proof}[Proof of \Cref{thm:secmain}]
By \Cref{lem:maximizer}, the first part of Step 1 returns a solution $\mx \in \bZfrac$ that satisfies,
	\begin{equation}\label{eq:ddtwo}
	\cphi \cdot \bg(\mx)\geq \expo{\bigO{-k \log n}} \max_{\bq \in \dsimplex}\probpml(\bq,\phi)~.
	\end{equation}
	We also sparsify the solution $\mx$ in Step 1 that we call $\mx'$. By \Cref{lem:sparse}, the solution $\mx' \in \bZfrac$ satisfies $\bg(\mx') \geq \bg(\mx)$ and $\big|\{i \in [\ell]~|~[\mx'\1]_{i}  > 0\}\big| \leq k+1$.
	The Steps 2-3 of our algorithm throw away the zero rows of matrix $\mx'$ and consider the sub matrix $\bS'$ corresponding to its non-zeros rows. Let $\bR'$ be the probability values that correspond to these non-zero rows of $\mx'$ and $\bS' \in \textbf{Z}^{\phi,frac}_{\bR'}$. As $\bS'$ changes during Steps 4-8 of the algorithm, we use $\my$ to denote the unchanged $\bS'$ from Step 2. The matrix $\my \in \textbf{Z}^{\phi,frac}_{\bR'}$ satisfies: $\bg(\my) = \bg(\mx') \geq \bg(\mx)$ and  has $\ell' \leq k+1$ rows. In the remainder of the proof we show that the distribution $\bp'$ outputted by our algorithm satisfies $\probpml(\bp',\phi) \geq \expo{-O((\rmax-\rmin)n +k \log (\ell n))} \cphi \cdot \bg(\my)$ that further combined with $\bg(\my) \geq \bg(\mx)$ and \Cref{eq:ddtwo} proves the theorem. Now recall the definition of $\bg(\my)$,
	\newcommand{\rlp}{\textbf{r}_{\ell'}}
	\begin{equation}
	\bg(\my)\defeq \exp\Big(\sum_{i\in[1,\ell'],j\in[0,k]}\left[\mc'_{ij}\my_{ij}-\my_{ij}\log\my_{ij}\right]+\sum_{i\in[1,\ell']}[\my\vones]_{i}\log[\my\vones]_{i}\Big)~,
	\end{equation}
%	\begin{equation}
%	\bg(\mx)\defeq \exp\Big(\sum_{i\in[1,\ell],j\in[0,k]}\left[\mc_{ij}\mx_{ij}-\mx_{ij}\log\mx_{ij}\right]+\sum_{i\in[1,\ell]}[\mx\vones]_{i}\log[\mx\vones]_{i}\Big)~,
%	\end{equation}
	where $\mc'_{ij}= \mj \log \ri'$. We refer to the linear term in $\my$ of function $\bg(\my)$ as the first term and the remaining entropy like terms as the second. We denote the elements of set $\bR'$ by $\ri'$ and let $\ro' < \dots \rlp'$. The Steps 4-8 of our rounding algorithm transfer the mass of $\bS'$ from lower probability value rows to higher ones while maintaining the integral row sum for the current row . Formally at iteration $i$, our algorithm takes the current fractional part of the $i$'th row sum ($[\bS'\vones]_{i}-\floor{[\bS'\vones]_{i}}$) and moves it to row $i+1$ (corresponding to higher probability value) by updating matrix $\bS'$. As the first term in function $\bg(\cdot)$ is strictly increasing in the values of $\ri'$, it is immediate that the final solution $\bSext$ satisfies,
	\begin{equation}\label{eq:first}
	\sum_{i\in[1,\ell'],j\in[0,k]}\mc'_{ij}\bSext_{ij} \geq \sum_{i\in[1,\ell'],j\in[0,k]}\mc'_{ij}\my_{ij}~.
	\end{equation}
	The movement of the mass between the rows happen within the same column, therefore $\bSext$ satisfies the column constraints, that is $[{\bSext}^{\top}\vones]_{j}=\phi_j$ for all $j \in [k]$. As $[\bSext\vones]_{i}=\floor{[\bS'\vones]_{i}}$ for all $i \in [1,\ell]$, we also have that all the row sums are integral. Therefore to prove the theorem all that remains is to bound the loss in objective corresponding to the second term for Steps 4-8 and analysis of Steps 9-11. 
	 
	In Steps 4-8 at iteration $i$, note that we move at most $1$ unit of mass ($\frac{\floor{[\bS'\vones]_{i}}}{[\bS'\vones]_{i}}$) from row $i$ to $i+1$. Therefore the updated matrix $\bS'$ after Step 6 satisfies $\sum_{j \in [0,k]}(\bS'_{i+1,j}-\my_{i+1,j}) \leq 1$. As $\bSext_{i+1,j}=\bS'_{i+1,j}\frac{\floor{\|\bS'_{i+1}\|_1}}{\|\bS'_{i+1}\|_1}$ we have $\sum_{j \in [0,k]}(\bS'_{i+1,j}-\bSext_{i+1,j}) \leq 1$ and further combined with the previous inequality we get $\sum_{j \in [0,k]}|\bSext_{i+1,j}-\my_{i+1,j}| \leq 1$ for all $i \in [1,\ell'-1]$. For the first row, we have $\bSext_{1,j}=\my_{1,j}\frac{\floor{\|\my_{1}\|_1}}{\|\my_{1}\|_1}$ which also gives $\sum_{j \in [0,k]}|\bSext_{1,j}-\my_{1,j}| \leq 1$. Therefore for all $i \in [1,\ell']$ the following inequality holds,
	\begin{equation}\label{eq:massbound}
	\sum_{j \in [0,k]}|\bSext_{i,j}-\my_{i,j}| \leq 1~.
	\end{equation}
	As the function $x\log x$ and $-x\log x$ are $O(\log n)$-Lipschitz when $x\in [\frac{1}{n^{10}},\infty]\cup \{0\}$ and all the terms where $\my_{i,j},[\my\vones]_{i},\bSext_{i,j},[\bSext\vones]_{i}$ take values less than $1/n^{10}$ contribute very little (at most $\exp(O(1/n^8))$) to the objective. Therefore by \Cref{eq:massbound} we get,
	\begin{equation}\label{eq:second}
	\sum_{i\in[1,\ell'],j\in[0,k]}\left(-\bSext_{ij}\log\bSext_{ij}\right)\geq \sum_{i\in[1,\ell'],j\in[0,k]}\left(-\my_{ij}\log\my_{ij}\right) - O(\ell'\log n)~,
	\end{equation}
	\begin{equation}\label{eq:third}
	\sum_{i\in[1,\ell']}[\bSext\vones]_{i}\log[\bSext\vones]_{i} \geq \sum_{i\in[1,\ell']}[\my\vones]_{i}\log[\my\vones]_{i}- O(\ell'\log n)~,
	\end{equation}
	where in the above inequalities we used the Lipschitzness of entropy and negative of entropy functions. Therefore Steps 4-8 of the algorithm outputs a solution $\bSext$ that along with other conditions also satisfies \Cref{eq:first,eq:second,eq:third}. Now observe that we are not done yet as the solution $\bSext$ might violate the distributional constraint $\sum_{i\in [1,\ell']}\ri' \|\bSext_{i}\|_{1} \leq 1$; to address this in Steps 9-10 we construct a new probability $\bRext$ where we scale down the probability values in $\bR'$ by $c=\sum_{i\in [1,\ell']}\ri' \|\bSext_{i}\|_{1}$. Such a scaling immediately ensures the satisfaction of the distributional constraint with respect to $\bRext$. As the row sums of $\bSext$ are integral and it satisfies all the column constraints as well, we have that $\bSext \in \textbf{Z}^{\phi}_{\bRext}$. Let $\ri''=\ri'/c$ be the probability values in set $\bRext$, then note that,
	 \begin{equation}\label{eq:four}
	 \begin{split}
	 \sum_{i\in[1,\ell'],j\in[0,k]}\mj\bSext_{ij}\log \ri''&=\sum_{i\in[1,\ell'],j\in[0,k]}\mj\bSext_{ij}\log \frac{\ri'}{c}\\
	 &=\sum_{i\in[1,\ell'],j\in[0,k]}\mc'_{i,j}\bSext_{ij}-\log c \sum_{i\in[1,\ell'],j\in[0,k]}\mj\bSext_{ij} \\
	 & = \sum_{i\in[1,\ell'],j\in[0,k]}\mc'_{i,j}\bSext_{ij} - \log c\sum_{j\in[0,k]}\mj\phi_j \\
	 &=\sum_{i\in[1,\ell'],j\in[0,k]}\mc'_{i,j}\bSext_{ij} - n\log c ~. 
	 \end{split}
	 \end{equation}
	 All that remains is to provide an upper bound on the value of $c$. Observe that, $c=\sum_{i\in [1,\ell']}\ri' \|\bSext_{i}\|_{1}=\sum_{i\in [1,\ell']}\ri'\|\my_{i}\|_{1}+ \sum_{i\in [1,\ell']}\ri'( \|\bSext_{i}\|_{1}-\|\my_{i}\|_{1}) \leq 1+\rmax-\rmin$, where in the last inequality we used $\my \in \textbf{Z}_{\bR'}^{\phi}$ and $\sum_{i\in [1,\ell']}( \|\bSext_{i}\|_{1}-\|\my_{i}\|_{1})=0$. Substituting the bound on $c$ back into \Cref{eq:four} we get,
	 \begin{equation}\label{eq:five}
	 \begin{split}
	 \sum_{i\in[1,\ell'],j\in[0,k]}\mj\bSext_{ij}\log \ri'' &= \sum_{i\in[1,\ell'],j\in[0,k]}\mc'_{i,j}\bSext_{ij} - n\log c \\
	 & \geq \sum_{i\in[1,\ell'],j\in[0,k]}\mc'_{i,j}\bSext_{ij}- O((\rmax-\rmin)n) ~.
	 \end{split} 
	 \end{equation}
	 Using \Cref{eq:first,eq:second,eq:third,eq:five}, the function value $\bg(\bSext)$ with respect to $\bRext$ satisfies,
	 \begin{equation}\label{eq:final}
	 \begin{split}
	 \bg(\bSext) & \geq \expo{- O(\rmax-\rmin)n-O(\ell'\log n)}\bg(\my)\\
	 & \geq \expo{- O(\rmax-\rmin)n-O(k\log n)}\bg(\my),
	 \end{split}
	 \end{equation}
	 where in the last inequality we used $\ell'\leq k+1$.
	 %	This process ensures three things: firstly $[\bSext\vones]_{i}=\floor{[\bS'\vones]_{i}}$ therefore the $i$'th row sum of $\bSext$ is integral. Secondly, the contribution of the $i$'th row of $\bSext$ to the second term of function $\bg(\bSext)$ satisfies: $\exp\Big(\sum_{j\in[0,k]}\left[-\bS'_{ij}\log\bS'_{ij}\right]+[\bS'\vones]_{i}\log[\bS'\vones]_{i}\Big)$
	 %	
	 %	the first term (linear term) of $\bg(\bS'')$ is strictly greater than that of $\bg(\mx)$, that is $\sum_{i\in[1,\ell'],j\in[0,k]}\mc_{ij}\mx'_{ij} \geq \sum_{i\in[1,\ell'],j\in[0,k]}\mc_{ij}\mx'_{ij}$. We next bound the change in the second term. It is immediate that $|[\mx'\vones]_i-[\mx\vones]_i| \leq 1$ for all $i \in [\ell]$, even further we also have $\sum_{j \in [0,k]}|\mx'_{ij}-\mx_{ij}| \leq 1$.
	 %	
	 %	
	 %	 As discussed earlier, the row sums of $\bS$ are integral as well. 
	 As $\bSext \in \bZext$, by \Cref{lem:associateddist} the associated distribution $\bp'$ satisfies $\probpml(\bp',\phi) \geq \exp(-O(k\log n)) \cphi \cdot \bg(\bSext)$. Further combined with \Cref{eq:final}, $\bg(\my) \geq \bg(\mx)$ and \Cref{eq:ddtwo} we get,
	 $$\probpml(\bp',\phi)\geq \expo{- O(\rmax-\rmin)n-O(k\log n)}\max_{\bq \in \dsimplex}\probpml(\bq,\phi)~.$$
	 
	 In the remainder we provide the analysis for the running time of our algorithm. By \Cref{thm:convrt} we can solve the convex optimization problem in Step 1 in time $\otilde(|\bR|k^2)$. By \Cref{lem:sparse}, the sub routine $\sparse$ can be implemented in time $\otilde(|\bR|k^{\omega})$ and all the remaining steps correspond to the low order terms; therefore the final run time of our algorithm is $\otilde(|\bR|k^{\omega})$ and we conclude the proof. 
\end{proof}
The above result holds for a general $\bR$ and we choose this set carefully to prove \Cref{thm:resultmaintwo}.

\begin{proof}[Proof of \Cref{thm:resultmaintwo}]
As the probability values lie in a restricted range, we just need to discretize the interval $[\ell,u]$. We choose the probability discretization set $\bR$ with parameters $\alpha=k/n$, $\rmax=u$, $\rmin=\ell$ and $|\bR|=O(\frac{n\log \frac{u}{\ell}}{k})$. By \Cref{lem:probdisc}, we have $\max_{\bq \in \dsimplex}\Prob{\bq, \phi} \geq \expo{-k-6}\Prob{\bp, \phi}$. Further combined with \Cref{thm:secmain}, we conclude our proof.
%Given a profile $\phi \sim p$  Therefore we just need to descritize the space $[\ell/n,u]/n$ and there the value of $\rmax$ and $\rmin$ are bounded. Choose the following discretization set we get.	Although this new rounding algorithm need not output an approximate PML distribution, it outputs a distribution $\bp'$ that with high probability satisfies $\probpml(\bp',\phi)\geq \exp(-\otilde(\ub-\lb))\probpml(\bp,\phi)$, where $\bp$ is the hidden distribution from which samples were drawn.
\end{proof}

\subsection{Notation and the General Framework}\label{app:notation}
%!TEX root = neurips_2020.tex

Here we provide all the definitions and description of the general framework for symmetric property estimation using the PseudoPML~\cite{CSS19pseudo,HO19}. We start by providing definitions of pseudo profile and PseudoPML distributions.
\newcommand{\Fsj}{\phis(j)}
\begin{defn}[$S$-pseudo Profile]
	For any sequence $y^{n} \in \bX^{n}$ and $S \subseteq \bX$, let $\setd\defeq \{ \bff(y^n,x) \}_{x \in S}$ be the set of distinct frequencies from $S$ and let $\eled_1,\eled_2,\dots, \eled_{|\setd|}$ be these distinct frequencies. The $S$-\emph{pseudo} profile of a sequence $y^n$ and set $S$ denoted by $\phis=\Phis(y^n)$ is a vector in $\Z^{|\setd|}$, where $\Fsj\defeq|\{x\in S ~|~\bff(y^n,x)=\eled_j \}|$ is the number of domain elements in $S$ with frequency $\eled_{j}$. We call $n$ the length of $\phis$ as it represents the length of the sequence $y^n$ from which the pseudo profile was constructed. Let $\Phisn$ denote the set of all $S$-pseudo profiles of length $n$.
\end{defn}
The probability of a $S$-pseudo profile $\phis \in \Phisn$ with respect to $\bp \in \simplex$ is defined as follows,
\begin{equation}\label{eqpml1}
\Pr(\bp,\phis)\defeq\sum_{\{y^n \in \bX^n~|~ \Phis(y^n)=\phis \}} \bbP(\bp,y^n),
\end{equation}
we use notation $\Pr$ instead of $\probpml$ to differentiate between the probability of a pseudo profile from the profile.
\begin{defn}[$S$-PseudoPML distribution]
	For any $S$-pseudo profile $\phis \in \Phisn$, a distribution $\bp_{\phis} \in \simplex$ is a $S$-\emph{PseudoPML} distribution if $\bp_{\phis} \in \argmax_{\bp \in \simplex} \probpml(\bp,\phis)$. Further, a distribution $\bpml \in \simplex$ is a $(\beta,S)$-\emph{approximate PseudoPML} distribution if $\probpml(\bpml,\phis)\geq \beta \cdot \probpml(\bp_{\phis},\phis)$.
\end{defn}

We next provide the description of the general framework from \cite{CSS19pseudo}. The input to this general framework is a sequence of $2n$ i.i.d sample denoted by $x^{2n}$ from an underlying hidden distribution $\bp$, a symmetric property of interest $\bff$ and a set of frequencies $\fsub$. The output is an estimate of $f(\bp)$ using a mixture of PML and empirical distributions.  
\begin{algorithm}[H]
	\caption{General Framework for Symmetric Property Estimation}\label{algpml}
	\begin{algorithmic}[1]
		\Procedure{Property estimation}{$x^{2n},\bff, \fsub$}
		\State Let $x^{2n}=(\xon,\xtn)$, where $\xon$ and $\xtn$ represent first and last $n$ samples of $x^{2n}$ respectively.
		\State Define $S\defeq \{y \in \bX~|~f(\xon,y) \in \fsub \}$.
		\State Construct profile $\phis$, where $\Fsj\defeq|\{y\in S ~|~\bff(\xtn,y)=j \}|$.
		\State Find a $(\beta,S)$-approximate PseudoPML distribution $\bpml$ and empirical distribution $\hat{\bp}$ on $\xtn$.
		\State \textbf{return} $\bff_{S}(\bpml)+\bff_{\bar{S}}(\hat{\bp}) + \text{correction bias with respect to }\bff_{\bar{S}}(\hat{\bp})$.
		\EndProcedure
	\end{algorithmic}
\end{algorithm}
We call the procedure of estimation using the above general framework as the PseudoPML approach.

\subsection{Proof of \Cref{lem:pseudomain} and the Implementation of General Framework}\label{app:pseduo}
%!TEX root = neurips_2020.tex

%\section{Efficient Computation of PseudoPML and its Application to Symmetric Property Estimation}\label{app:universal}
Here we provide the proof of \Cref{lem:pseudomain}. The main idea behind the proof of this lemma is to use an efficient solver for the computation of approximate PML to return an approximate PseudoPML distribution. The following lemma will be useful to establish such a connection and we define the following notations: $\Delta^{S}_{[\ell,u]}\defeq \{\bp\in \Delta^{S}\Big|\bp_x \in [\ell,u] ~\forall x\in S\}$ and further define $\Delta^{\bX}_{S,[\ell,u]}\defeq \{\bp\in \Delta^{\bX}\Big|\bp_x \in [\ell,u] ~\forall x\in S\}$, where $\Delta^{S}$ are all distributions that are supported on domain $S$.

\newcommand{\ns}{n'}
\newcommand{\nsbar}{n''}
\newcommand{\yns}{y^{\ns}}
\newcommand{\ynsbar}{y^{\nsbar}}

%where $\Delta^{S}_{[\ell,u]}\defeq \{\bp\in \Delta^{S}\Big|\bp_x \in [\ell,u] ~\forall x\in S\}$
\begin{lemma}\label{lem:computepseudo}
	For any profile $\phi' \in \Phi^{n'}$ with $k'$ distinct frequencies, domain $S \subset \bX$ and $\ell',u' \in [0,1]$.
	% and $\bp$ be a distribution where all its probability values lie in the interval $[\ell,u]$. 
	If there is an algorithm that runs in time $T(n',k',u',\ell')$ and returns a distribution $\bp' \in \Delta^{S}$ such that,
	$$\probpml(\bp',\phi') \geq \expo{-O((u'-\ell')n'\log n'+k'\log n')} \max_{\bq \in \Delta^{S}_{[\ell,u]}}\probpml(\bq,\phi')~.$$
	Then for domain $\bX$, any pseudo $\phis \in \Phisn$ with $k$ distinct frequencies and $\ell,u \in [0,1]$, such an algorithm can be used to compute $\bp''_{S}$, part corresponding to $S\subseteq \bX$ of distribution $\bp'' \in \simplex
	$ in time $T(n,k,u,\ell)$ where the distribution $\bp''$ further satisfies,
	$$\Pr(\bp'',\phis) \geq \expo{-O((u-\ell)n\log n+k\log n)} \max_{\bq \in \simplex_{S,[\ell,u]}}\Pr(\bq,\phis)~.$$
	\end{lemma}
\begin{proof}
	\renewcommand{\ns}{n_1}
	\renewcommand{\nsbar}{n_2}
Recall that, 
$$\Pr(\bq,\phis)\defeq\sum_{\{y^n \in \bX^n~|~ \Phis(y^n)=\phis \}} \bbP(\bq,y^n)~.$$
Let $\bq_{S}$ and $\bq_{\bar{S}}$ denote the part of distribution $\bq$ corresponding to $S,\bar{S} \subseteq \bX$; they are pseudo distributions supported on $S$ and $\bar{S}$ respectively. Let $\ns= \sum_{\mj \in \phis} \mj$ and $\nsbar \defeq \sum_{\mj \in \phi_{\bar{S}}} \mj$ then, 
$$\probpml(\bq_{S},\phi_{S}) \defeq \sum_{\{y^{\ns} \in S^{\ns}~|~\Phi(\yns) = \phis\}} \prod_{x \in S} \bq_{x}^{\bff(\yns,x)}$$
$$\probpml(\bq_{\bar{S}},\phi_{\bar{S}}) \defeq \sum_{\{y^{\nsbar} \in {\bar{S}}^{\nsbar}~|~\Phi(\ynsbar) = \phi_{\bar{S}}\}} \prod_{x \in \bar{S}} \bq_{x}^{\bff(\ynsbar,x)}$$
We can write the probability of a pseudo profile in terms of the above functions as follows,
$$\Pr(\bq,\phis) = \probpml(\bq_{S},\phi_{S}) \probpml(\bq_{\bar{S}},\phi_{\bar{S}}).$$
Therefore,
\begin{align*}
\max_{\bq \in \simplex}\Pr(\bq,\phis) = \max_{\bq \in \simplex} \probpml(\bq_{S},\phi_{S}) \probpml(\bq_{\bar{S}},\phi_{\bar{S}})~,
\end{align*}
In the applications of PseudoPML, we just require the part of the distribution corresponding to $S\subseteq \bX$ and in the remainder we focus on its computation by exploiting the product structure in the objective. 
$$\max_{\bq \in \simplex} \probpml(\bq_{S},\phi_{S}) \probpml(\bq_{\bar{S}},\phi_{\bar{S}}) = \max_{\alpha \in [0,1]} \left(\alpha^{\ns} \max_{\bq' \in \Delta^{S}} \probpml(\bq',\phi_{S})\right) \left((1-\alpha)^{\nsbar}\max_{\bq'' \in \Delta^{\bar{S}}} \probpml(\bq'',\phi_{\bar{S}}) \right),$$
where in the above objective we converted the terms involving the pseudo distributions to distributions. The above equality holds because scaling all the probability values of a distribution by a factor of $\alpha$ scales the PML objective by a factor of $\alpha$ to the power of length of the profile, which is $\ns$ and $\nsbar$ for $\phis$ and $\phi_{\bar{S}}$ respectively.
The above objective is nice as we can just focus on the first term in the objective corresponding to $S$ given the optimal $\alpha$ value. Note in the above optimization problem the terms $\max_{\bq' \in \Delta^{S}} \probpml(\bq',\phi_{S})$ and $\max_{\bq'' \in \Delta^{\bar{S}}} \probpml(\bq'',\phi_{\bar{S}})$ are independent of $\alpha$ and we can solve for the optimum $\alpha$ by finding the maximizer of the following optimization problem.
$$\max_{\alpha \in [0,1]} \alpha^{\ns} (1-\alpha)^{\nsbar}~.$$ 
The above optimization problem has a standard closed form solution and the optimum solution is  $\alpha^{*}=\frac{\ns}{\ns+\nsbar}=\frac{\ns}{n}$.
To summarize, the part of distribution $\bp''$ corresponding to $S$ that satisfies the guarantees of the lemma can be computed by solving the optimization problem $\max_{\bq' \in \Delta^{S}} \Pr(\bq',\phi_{S})$ upto multiplicative accuracy of $ \expo{-O((u-\ell)n\log n+k\log n)}$ and then scaling all the entries of the corresponding distribution supported on $S$ by a factor of $\ns/n$; which by the conditions of the lemma can be computed in time $T(n,k,\ell,u)$ and we conclude the proof.
\end{proof}
Using the above lemma we now provide the proof for \Cref{lem:pseudomain}.
\newcommand{\issimplex}{\simplex_{S,I}}
\newcommand{\ssimplex}{\Delta^{S}}
\begin{proof}[Proof of \Cref{lem:pseudomain}]
	Let $\bp,\bpml$ be the underlying hidden distribution and $(\beta,S)$-approximate PseudoPML distribution. The guarantees stated in the lemma are the efficient version of Theorem 3.9 and 3.10 in \cite{CSS19pseudo}. Both these theorems are derived using Theorem 3.8 in \cite{CSS19pseudo} that in turn depends on Theorem 3.7 which captures the performance of an approximate PseudoPML distribution. In all these proofs the only expression where the definition of $(\beta,S)$-approximate PseudoPML distribution was used is the following: $\prob{\bpml,\phis} \geq \beta \prob{\bp,\phis}$. Any other distribution $\bp'$ that satisfies $\prob{\bp',\phis} \geq \beta \prob{\bp,\phis}$ also has the same guarantees and provides the efficient version of Theorem 3.9 and 3.10, that is the guarantees of our lemma.
	
	As described in \Cref{app:notation}, the general framework works in two steps. In the first step, it takes the first half of the samples ($x_1^n$) and determines the set $S\defeq \{y \in \bX~|~f(\xon,y) \in \fsub \}$, where $\fsub$ is a predetermined subset of frequencies (input to the general framework) that depends on the property of interest. The pseudo profile $\phis$ is computed on the second half of the samples, that is $\Fsj\defeq|\{y\in S ~|~\bff(\xtn,y)=j \}|$. Based on the frequency of the elements of $S$ in the first half of the sample (they all belong to $\fsub$), with high probability (in the number of samples) we have an interval $I=[\ell,u]$ in which all the probability values of elements in $S \subseteq \bX$ for $\bp$ lie. Therefore finding a distribution $\bp'$ that satisfies,
	$$\prob{\bp',\phis} \geq \beta \max_{\bq \in \issimplex} \prob{\bq,\phis} \implies \prob{\bp',\phis} \geq \beta \prob{\bp,\phis}~,$$
	 where $\issimplex \defeq \{\bq \in \simplex~\Big|~ \bq_{x} \in I \text{ for all }x\in S\}$; therefore $\bp'$ can be used as a proxy for $\bpml$ and both these distributions satisfy the guarantees of our lemma (for entropy and distance to uniformity) for an appropriately chosen $\beta$. The value of $\beta$ depends on the size of $\fsub$ that further depends on the property of interest and we analyze this parameter for each property in the final parts of the proof. 
	 
	 Now note that we need to find a distribution $\bp'$ that satisfies, $\prob{\bp',\phis} \geq \beta \max_{\bq \in \issimplex} \prob{\bp,\phis}$ and to implement the PseudoPML approach all we need is $\bp'_{S}$, the part of the distribution corresponding to $S$. The \Cref{lem:computepseudo} helps reduce the problem of computing PseudoPML to PML and we use the algorithm given to us by the condition of our lemma to compute $\bp'_{S}$.
%	 we just need to find a distribution supported on $S\subseteq \bX$ that satisfies,
%	 $$\probpml(\bp'_{S},\phis) \geq \probpml(\bp_{S},\phis)~.$$
%	 In the above we reduced the problem of computing PseudoPML to PML and we use the algorithm provided in the condition of the lemma to find such a distribution. The running time of the algorithm depends on $\ell,u$ where $I=[\ell,u]$

	 In the remainder, we study the running time and the value of $\beta$ for entropy and distance to uniformity.
	
	\paragraph{Entropy:} In the application of general framework (\Cref{algpml}) to entropy, the authors in \cite{CSS19pseudo} choose $F=[0,c \log n]$, where $c>0$ is a fixed constant (See proof of Theorem 3.9 in \cite{CSS19pseudo}). Recall the definition of subset $S\defeq \{y \in \bX~|~f(\xon,y) \in \fsub \}$ and as argued in the proof of Theorem 3.9 in \cite{CSS19pseudo}, with high probability all the domain elements $x\in S$ have probability values $\bp_{x} \leq \frac{2c \log n}{n}$. Further, we can assume that the minimum non-zero probability of distribution $\bp$ to be $\Omega(1/\poly(n))$, because in our setting $n \in \Omega(N/\log N)$ for all error parameters $\epsilon$ and the probability values less than $1/\poly(n)$ contribute very little to the probability mass or entropy of the distribution and we can ignore them. Therefore to implement the PseudoPML approach for entropy all we need is the part corresponding to $S$ of distribution $\bp'$ that satisfies,
	\begin{equation}\label{eq:satisfy}
	\prob{\bp',\phis} \geq \beta \max_{\bq \in \issimplex} \prob{\bq,\phis}~,
	\end{equation}
	for any $\beta > \expo{-O(\log^2 n)}$ (Theorem 3.9 in \cite{CSS19pseudo}) and $I=[\frac{1}{\poly(n)},\frac{2c \log n}{n}]$. 
	Based on our discussion at the start of the proof, this corresponds to computing the $\beta$-approximate PML distribution supported on $S$ for the profile $\phis$. As the number of distinct frequencies in the profile $\phis$ is at most $O(\log n)$, length of the profile $\phis$ is at most $n$ and interval $I=[\ell,u]$ take values $\ell=1/\poly(n)$ and $u=O(\frac{\log n}{n})$, the algorithm given by the conditions of our lemma computes the part corresponding to $S$ of distribution $\bp'$ that satisfies \Cref{eq:satisfy} with approximation factor $\beta > \expo{-O(\log^2 n)}$ in time $T(n,O(\log n),1/\poly(n),O(\frac{\log n}{n}))$.

	The proof for distance to uniformity is similar to that of entropy and is described below.
	
	\paragraph{Distance to Uniformity:} For distance to uniformity, the authors in \cite{CSS19pseudo} choose $F=[\frac{n}{N}-\sqrt{\frac{c n \log n}{N}}, \frac{n}{N}+\sqrt{\frac{c n \log n}{N}}]$, where $c$ is a fixed constant (See proof of Theorem 3.10 in \cite{CSS19pseudo}). The subset $S\defeq \{y \in \bX~|~f(\xon,y) \in \fsub \}$ and as argued in the proof of Theorem 3.10 in \cite{CSS19pseudo}, with high probability all the domain elements $x\in S$ have probability values $\bp_{x} \in [\frac{1}{N}-\sqrt{\frac{2c \log n}{nN}}, \frac{1}{N}+\sqrt{\frac{2c \log n}{nN}}]$. Therefore to implement the PseudoPML approach for distance to uniformity all we need is the part corresponding to $S$ of distribution $\bp'$ that satisfies,
	\begin{equation}\label{eq:satisfy1}
	\prob{\bp',\phis} \geq \beta \max_{\bq \in \issimplex} \prob{\bq,\phis}~,
	\end{equation}
	for any $\beta > \expo{-O(\sqrt{\frac{c n \log^3 n}{N}}}$ (Theorem 3.10 in \cite{CSS19pseudo}) and $I=[\frac{1}{N}-\sqrt{\frac{2c \log n}{nN}}, \frac{1}{N}+\sqrt{\frac{2c \log n}{nN}}]$. This corresponds to computing the $\beta$-approximate PML distribution supported on $S$ for the profile $\phis$. As the number of distinct frequencies in the profile $\phis$ is at most $\sqrt{\frac{2c n \log n}{N}} \in O(1/\epsilon)$ (because $n=\Theta(\frac{N}{\epsilon^2\log N})$ for distance to uniformity), length of the profile $\phis$ is at most $n$ and interval $I=[\ell,u]$ take values $\ell=\frac{1}{N}-\sqrt{\frac{2c \log n}{nN}} \in \Omega(1/N)$ and $u=\frac{1}{N}+\sqrt{\frac{2c \log n}{nN}} \in O(1/N)$, the algorithm given by the conditions of our lemma computes the part corresponding to $S$ of distribution $\bp'$ that satisfies \Cref{eq:satisfy1} with approximation factor $\beta > \expo{-O(\sqrt{\frac{c n \log^3 n}{N}}}$ in time $T(n,O(1/\epsilon),\Omega(1/N),O(1/N)$. We conclude the proof.
\end{proof}

\subsection{Experiments}\label{app:experiments}
%In this section, we provide details related to PseudoPML implementation and some additional experiments. In these additional experiments, we compare the performance of the PseudoPML approach implemented using our approximate PML algorithm (\Cref{alg:round2}) with a heuristic algorithm~\cite{PJW17}.

In this section, we provide details related to PseudoPML implementation and some additional experiments. We perform different sets of experiments for entropy estimation -- first to compare performance guarantees of PseudoPML approach implemented using our rounding algorithm to the other state-of-the-art estimators and the other to compare the performance of the PseudoPML approach implemented using our approximate PML algorithm (\Cref{alg:round2}) with a heuristic algorithm~\cite{PJW17}.

All the plots in this section depict the performance of various algorithms for estimating entropy of different distributions with domain size $N=10^5$. Each data point represents 50 random trials. ``Uniform'' is the uniform distribution, ``Mix 2 Uniforms'' is a mixture of two uniform distributions, with half the probability mass on the first $N/10$ symbols and the remaining mass on the last $9N/10$ symbols, and $\mathrm{Zipf}(\alpha) \sim 1/i^{\alpha}$ with $i \in [N]$. In the PseudoPML implementation for entropy, we divide the samples into two parts. We run the empirical estimate on one (this is easy) and the PML estimate on the other. Similar to \cite{CSS19pseudo}, we pick $threshold=18$ (same as \cite{WY16}) to divide the samples, i.e. we use the PML estimate on frequencies $\leq 18$ and empirical estimate on the rest. As in \cite{CSS19pseudo}, we do not perform sample splitting. In all the plots, ``Our work'' corresponds to the implementation of this PseudoPML approach using our second approximate PML algorithm presented in \Cref{sec:practical} (\Cref{alg:round2}). Refer to \cite{CSS19pseudo} for further details on the PseudoPML approach.

In \Cref{fig2}, we compare performance guarantees of our work to the other state-of-the-art estimators for entropy. We already did this comparison in \Cref{subsec:expirements} and here we do it for three other distributions. As described in \Cref{subsec:expirements}, MLE is the naive approach of using the empirical distribution with correction bias; all the remaining algorithms are denoted using bibliographic citations.
\begin{figure}[ht]
	\begin{center}
		\begin{tabular}{ccc}
			\begin{overpic}[width=.3\columnwidth]{%,grid]{%
					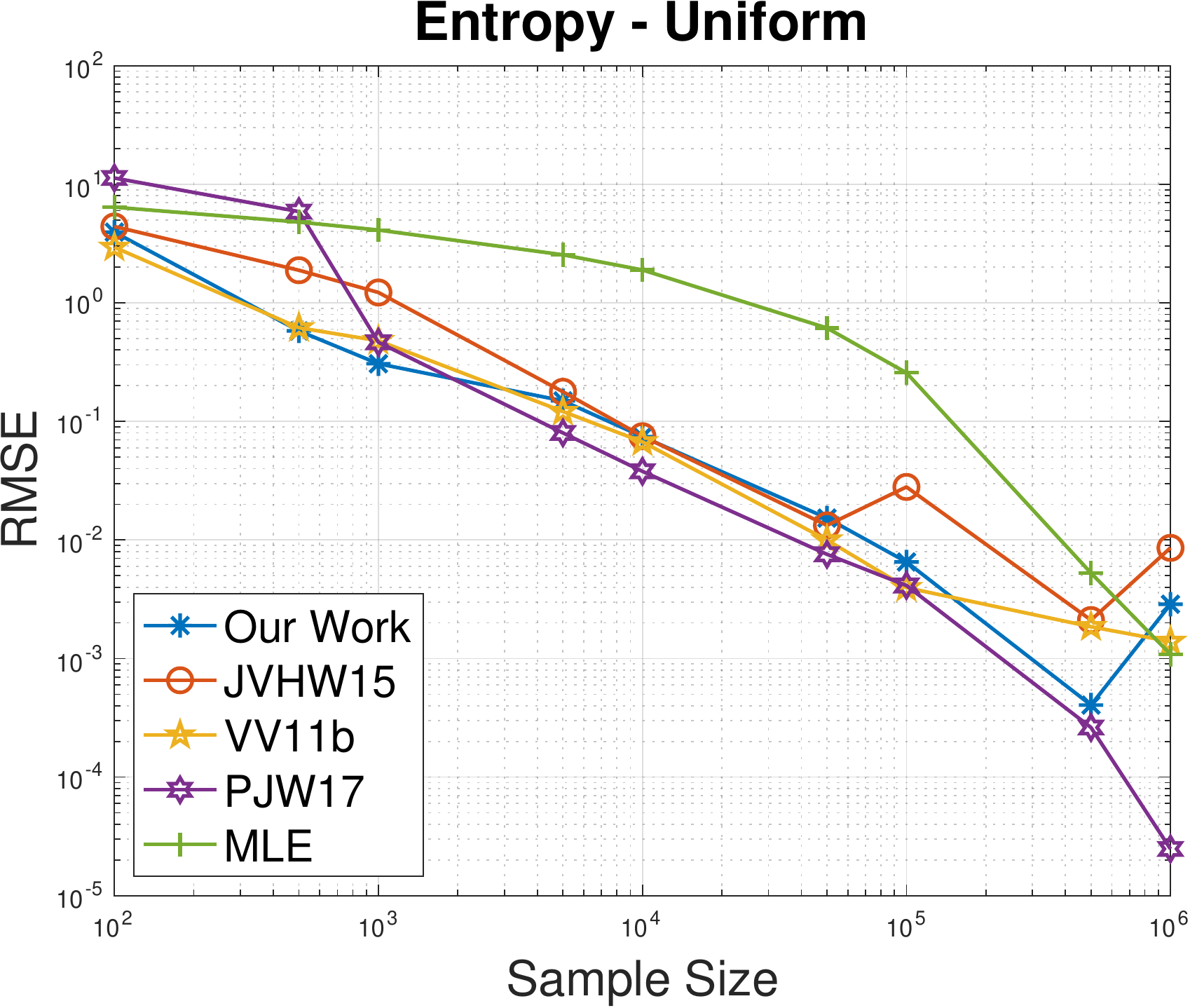}
			\end{overpic} &
			\begin{overpic}[width=.3\columnwidth]{%,grid]{%
					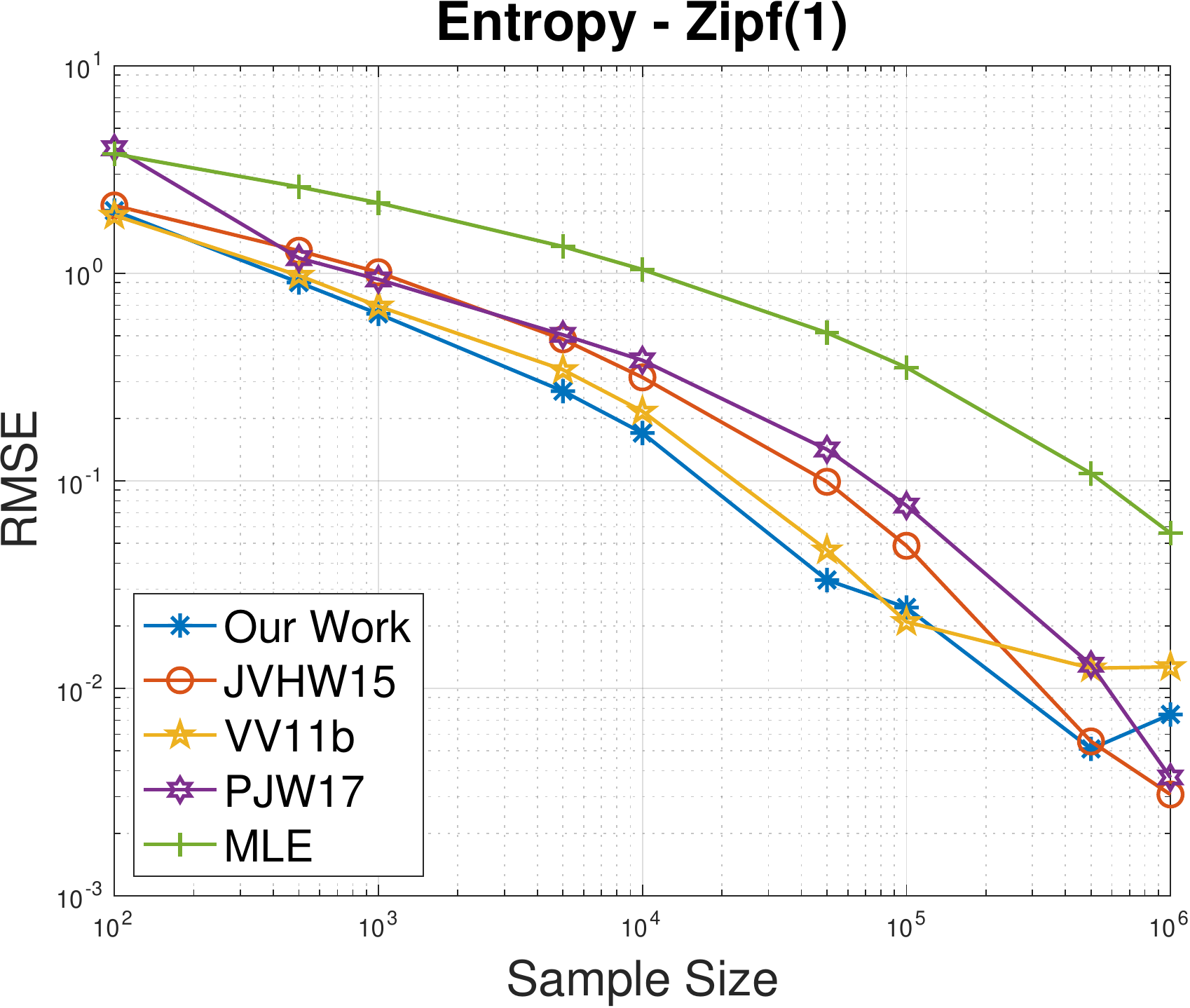}
			\end{overpic} &
			\begin{overpic}[width=.3\columnwidth]{%,grid]{%
					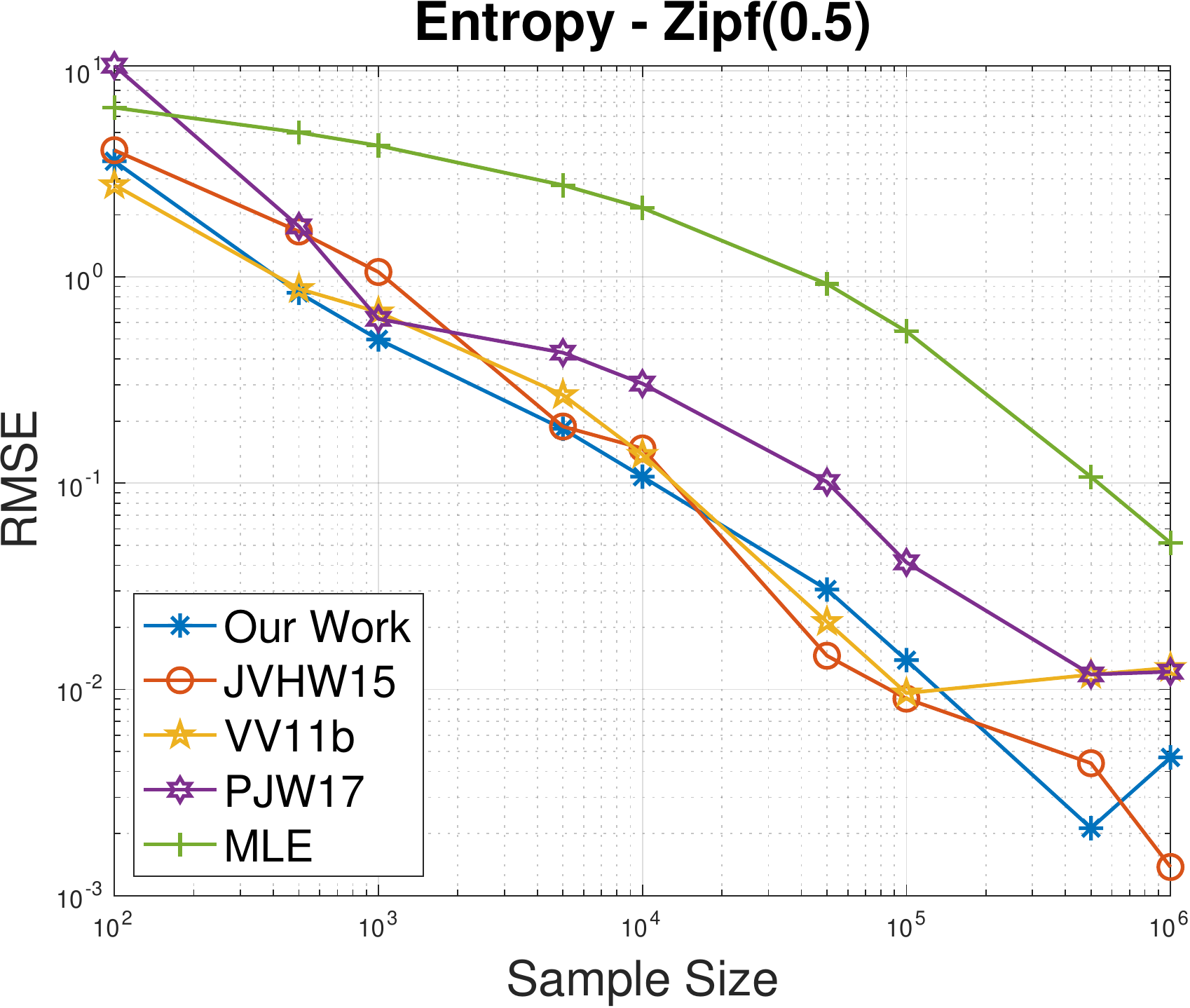}
			\end{overpic}
			%			(a) & (b) & (c)
		\end{tabular}
		\caption{
			\label{fig2}
			Experimental results for entropy estimation.
		}
	\end{center}
\end{figure}

An advantage of the pseudo PML approach is that it one can use any algorithm to compute the part corresponding to the PML estimate as a black box. In \Cref{fig3}, we perform additional experiments for six different distributions comparing the PML estimate computed using our algorithm (``Our work'') versus the algorithm in \cite{PJW17} (``Pseudo-PJW17''), a heuristic approach to compute the approximate PML distribution.
%%enhances the usage of these algorithms by 
%providing both competitive performance and running time efficiency. 
%In our experiments, we use the heuristic algorithm in \cite{PJW17} to compute an approximate PML distribution. In the first set of experiments detailed below, we compare the performance of the pseudo PML approach with raw \cite{PJW17} and other state-of-the-art estimators for estimating entropy. Our code is available at \url{https://github.com/shiragur/CodeForPseudoPML.git}

%MLE is the naive approach of using the empirical distribution with correction bias;
%%, PJW is \cite{PJW17}, VV is \cite{VV11a}, and JVHW is \cite{JVHW15}
%all the remaining algorithms are denoted using bibliographic citations.
%%\footnote{[PJW17]~~ D. S. Pavlichin, J. Jiao, and T. Weissman. Approximate Profile Maximum Likelihood.
%In our algorithm we pick $threshold=18$ (same as \cite{WY16}) and our set $\mathrm{F}=[0,18]$ (input of \Cref{algpml}), i.e. we use the PML estimate on frequencies $\leq 18$ and empirical estimate on the rest. Unlike \Cref{algpml}, we do not perform sample splitting in the experiments -- we believe this requirement is an artifact of our analysis. For estimating entropy, the error achieved by our estimator is competitive with \cite{PJW17} and other state-of-the-art entropy estimators. Note that our results match \cite{PJW17} for small sample sizes because not many domain elements cross the threshold and for a large fraction of the samples, we simply run the \cite{PJW17} algorithm.

\begin{figure}[ht]
	\begin{center}
		\begin{tabular}{ccc}
			\begin{overpic}[width=.3\columnwidth]{%,grid]{%
					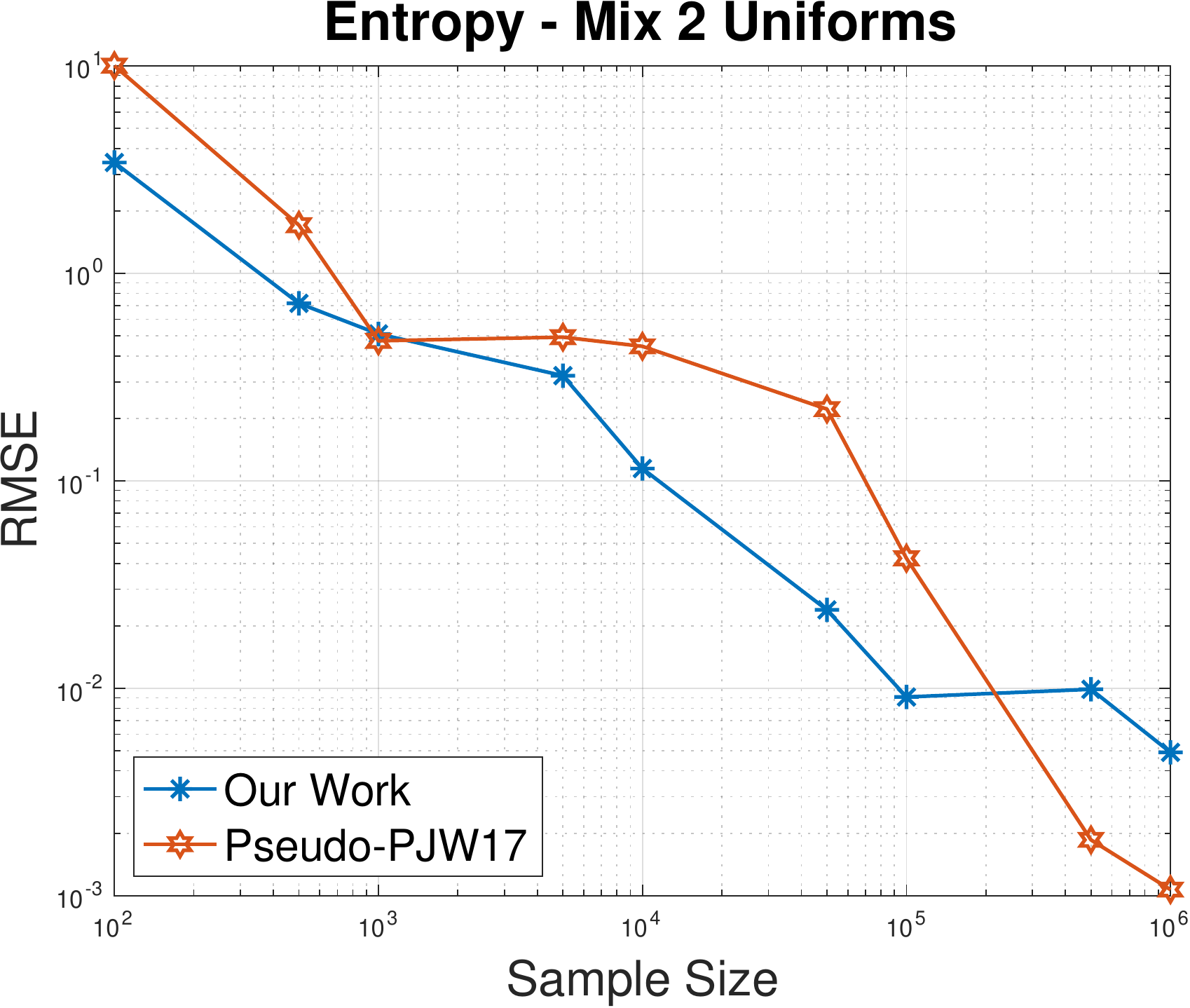}
			\end{overpic} &
			\begin{overpic}[width=.3\columnwidth]{%,grid]{%
					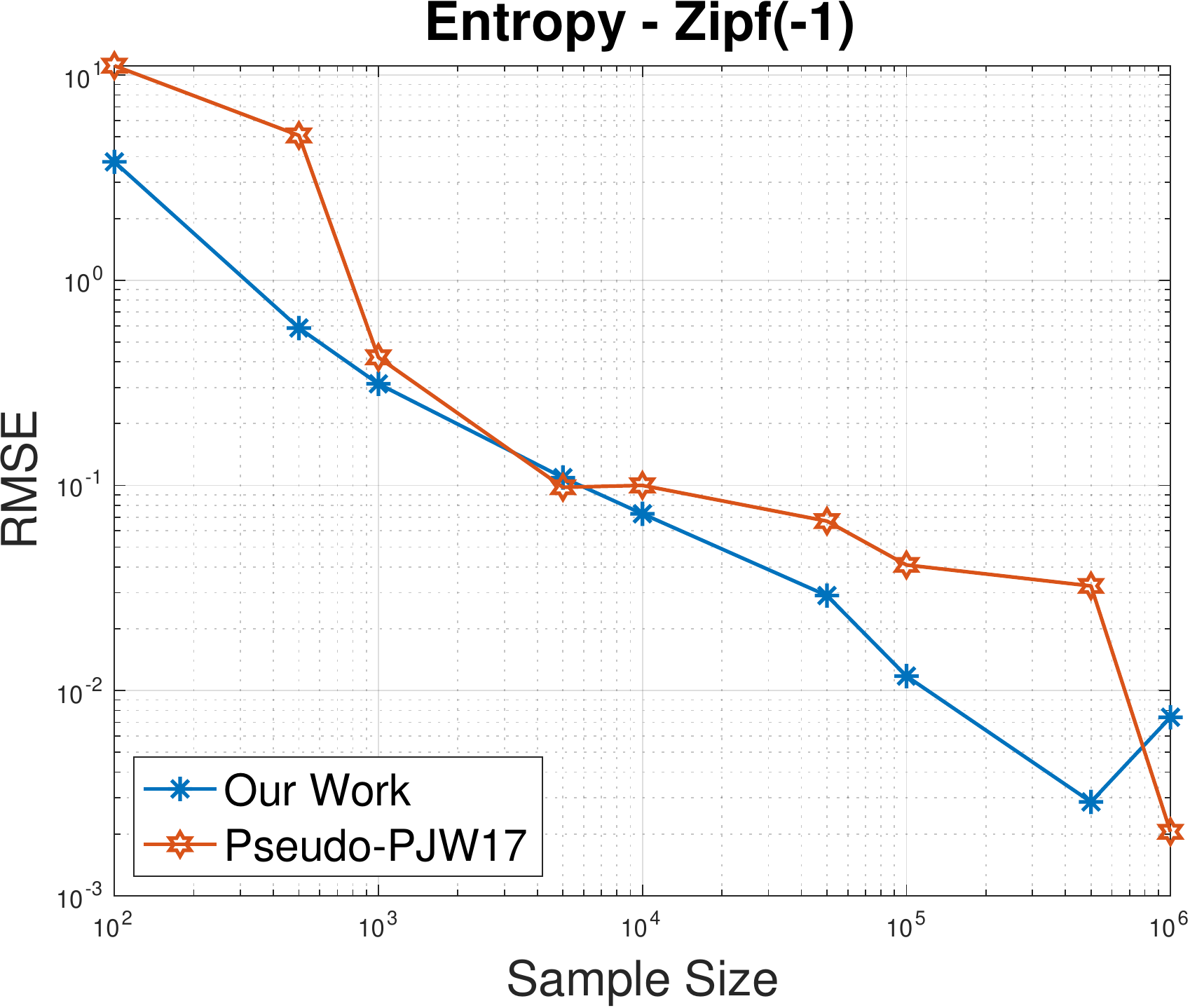}
			\end{overpic} &
			\begin{overpic}[width=.3\columnwidth]{%,grid]{%
					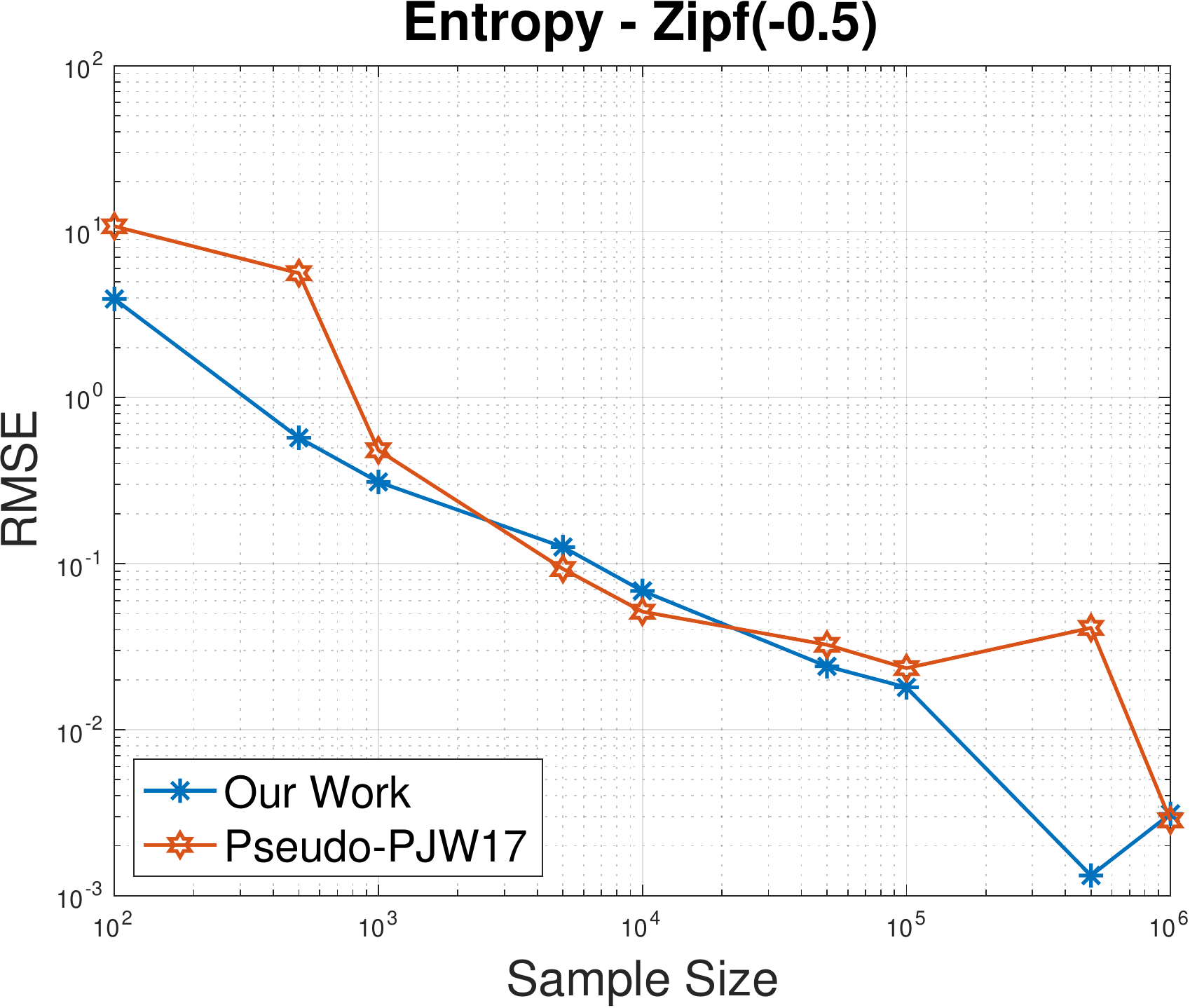}
			\end{overpic}
			%			(a) & (b) & (c)
		\end{tabular}
	\begin{tabular}{ccc}
		\begin{overpic}[width=.3\columnwidth]{%,grid]{%
				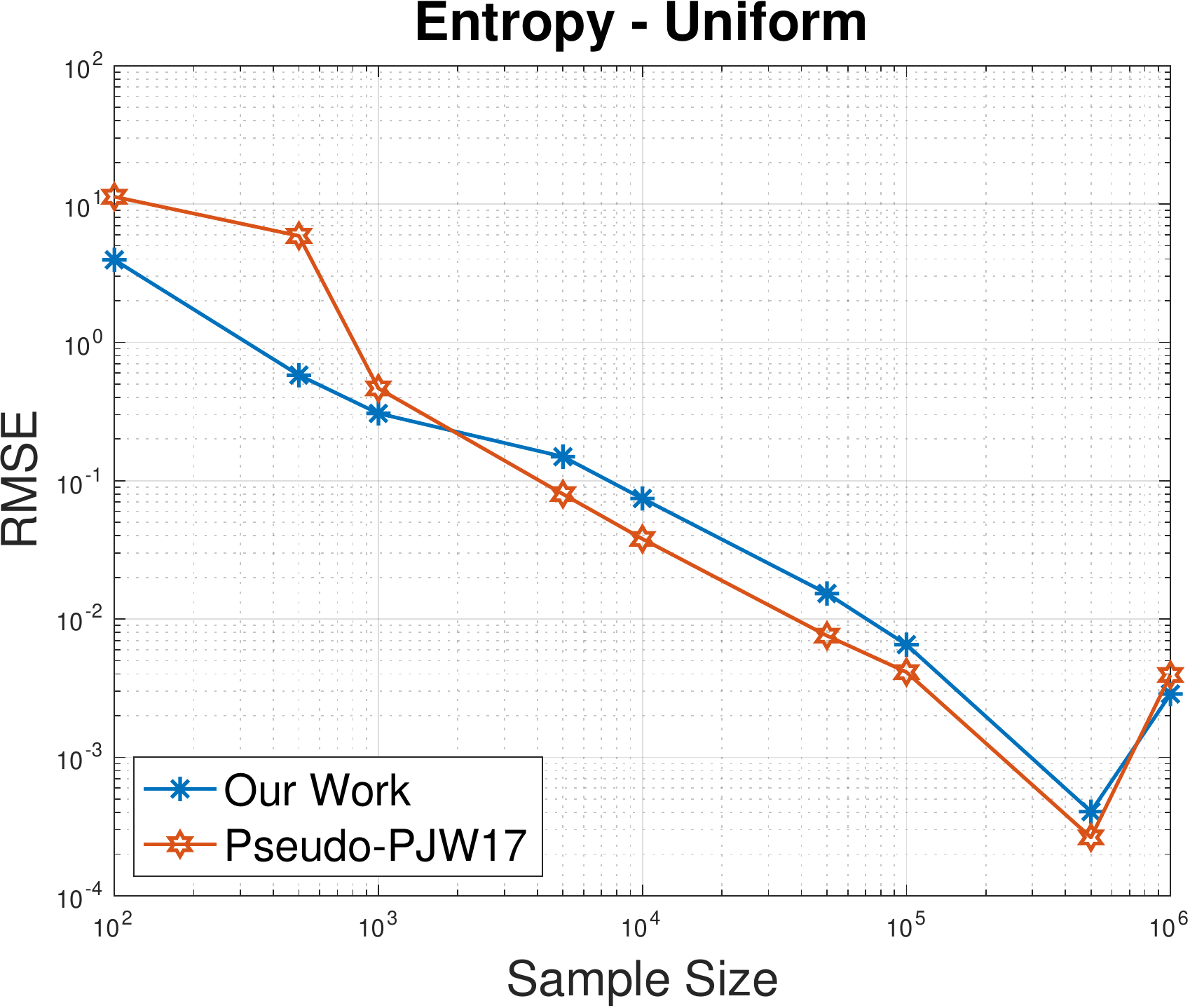}
		\end{overpic} &
		\begin{overpic}[width=.3\columnwidth]{%,grid]{%
				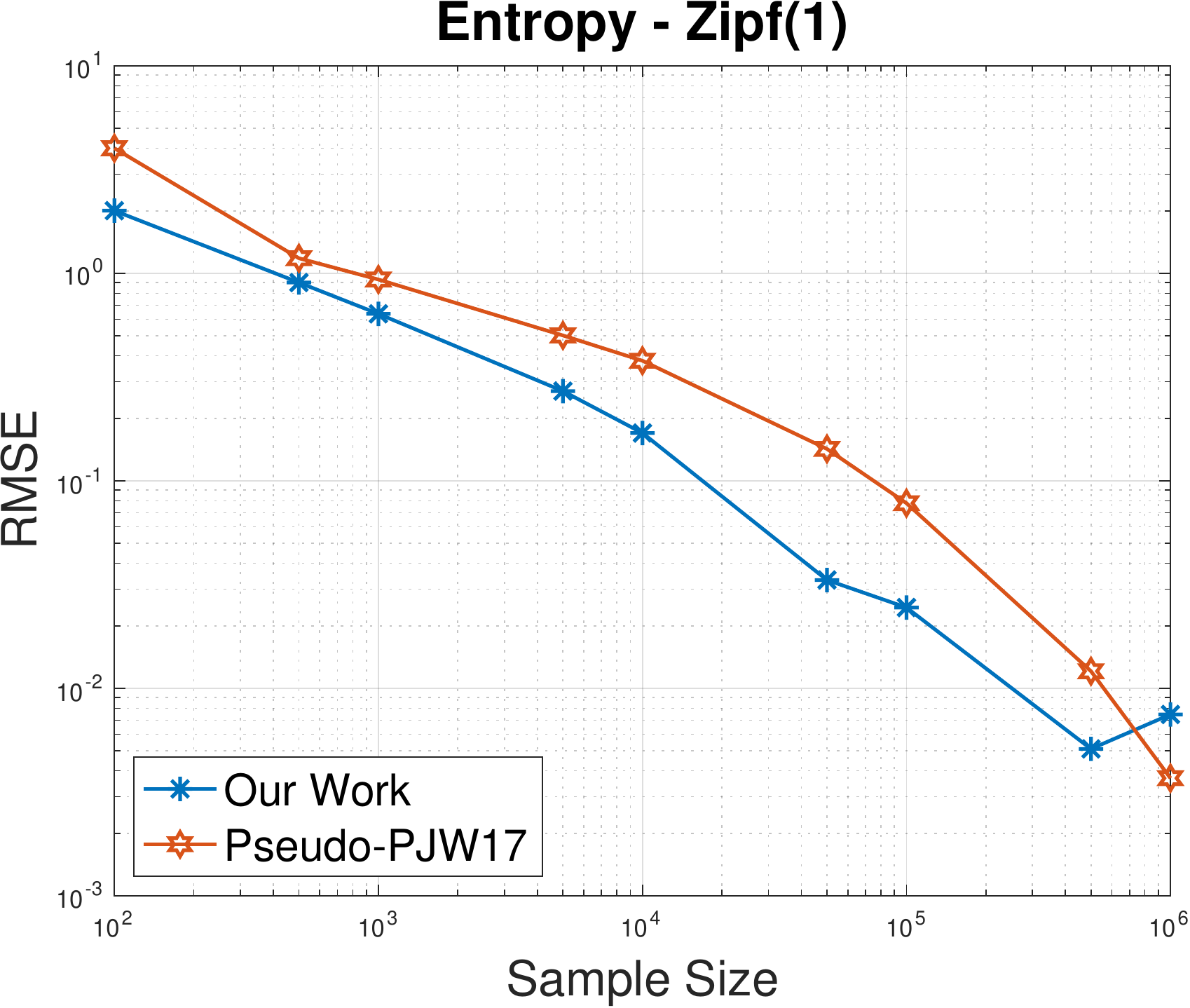}
		\end{overpic} &
		\begin{overpic}[width=.3\columnwidth]{%,grid]{%
				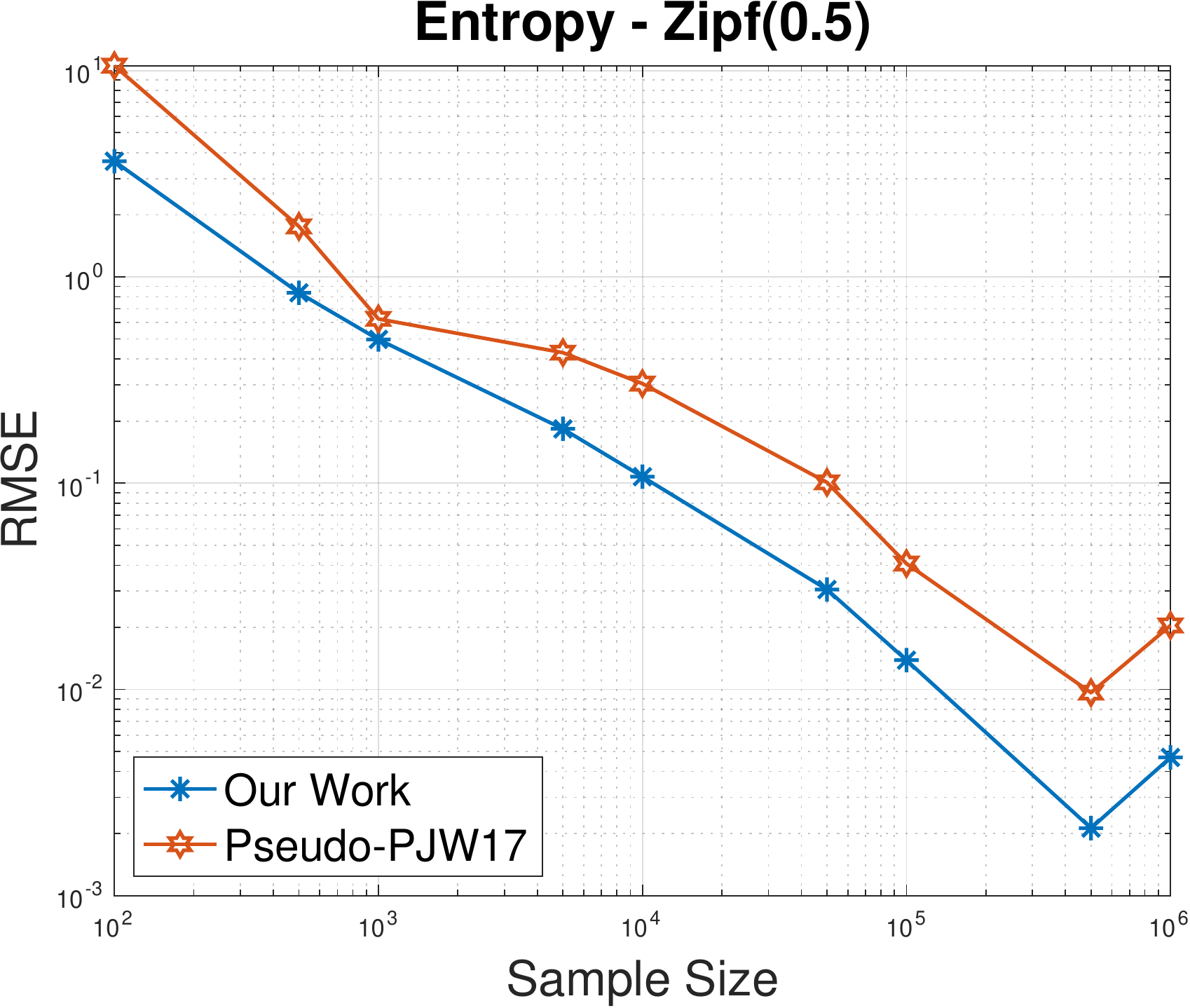}
		\end{overpic}
		%			(a) & (b) & (c)
	\end{tabular}
		\caption{
			\label{fig3}
			Experimental results for entropy estimation.
		}
	\end{center}
\end{figure}

In the remainder we provide further details on the implementation of our algorithm (\Cref{alg:round2}). In Step 1, we use CVX\cite{cvx} with package CVXQUAD\cite{FSP17} to solve the convex program. The accuracy of discretization determines the number of variables in the convex program and for practical purposes we perform very coarse discretization which reduces the number of variables to our convex program and helps implement Step 1 faster. The size of the discretization set we choose is slightly more than the number of distinct frequencies. Even with such coarse discretization, we still achieve results that are comparable to the other state-of-the-art entropy estimators. The intuition behind to choice of such a discretization set is because of \Cref{lem:sparse}, which guarantees the existence of a sparse solution. As the discretization set is already of small size, we do not require to perform further scarification and we avoid invoking the $\sparse$ subroutine; therefore providing a faster practical implementation. 
%\sidford{Might be good to say the exact values and how set accuracy to solve and stuff like that.}

%\section{Details on Profile Entropy \cite{HO20}}\label{app:profileentropy}
%\input{appendix_proofs.tex}
%\input{dual.tex}
\end{document}